\documentclass[10pt]{article}
\pdfoutput=1
\usepackage[latin1]{inputenc}							
\usepackage[T1]{fontenc}								
\usepackage{dsfont,amsmath,amsfonts,amssymb,mathtools,graphicx,subfigure,amsthm,braket,color,tikz,setspace,palatino,euler}
\usepackage[longnamesfirst,numbers,square]{natbib}			
\usepackage[total={6.5in,8.75in}, top=1.2in, includefoot]{geometry}
\usepackage[pdfpagemode=UseNone, pdfstartview={XYZ null null null}]{hyperref}

\definecolor{Yvan}{rgb}{0.2, 0.6, 0.24}
\definecolor{Guiom}{rgb}{0.25, 0.24, 0.6}
\definecolor{dark-red}{rgb}{0.4,0.15,0.15}
\definecolor{dark-blue}{rgb}{0.15,0.15,0.4}
\definecolor{medium-blue}{rgb}{0,0,0.5}
\hypersetup{colorlinks, linkcolor={dark-red}, citecolor={dark-blue}, urlcolor={medium-blue}}

\newcommand{\cn}{\otimes^n\mathbb{C}^2}
\newcommand{\tl}{\mathrm{TL}_n}
\newcommand{\tlsans}{\mathrm{TL}}
\newcommand{\uq}{\mathrm{U}_q\mathfrak{sl}_2}
\newcommand{\hn}{\mathrm{H}_n}
\newcommand{\C}{\mathbb{C}}
\newcommand{\N}{\mathbb{N}}
\newcommand{\Z}{\mathbb{Z}}

\newcommand{\qbin}[2]{\begin{bmatrix} #1 \\ #2 \end{bmatrix}}
\newcommand{\qbinmini}[2]{\left[\begin{smallmatrix} #1 \\ #2 \end{smallmatrix}\right]}
\newcommand{\bin}[2]{\begin{pmatrix} #1 \\ #2 \end{pmatrix}}
\newcommand{\binmini}[2]{\left(\begin{smallmatrix} #1 \\ #2 \end{smallmatrix}\right)}
\newcommand{\modulo}{\ \textrm{mod}\ }

\theoremstyle{plain}
\newtheorem{thm}{Theorem}[section]
\newtheorem{lem}[thm]{Lemma}
\newtheorem{prop}[thm]{Proposition}
\newtheorem{cor}{Corollary}[section]

\theoremstyle{definition}
\newtheorem{defn}{Definition}[section]

\theoremstyle{remark}

\title{\huge The idempotents of the $\tl$-module $\cn$\\ in terms of elements of $\uq$}
\author{
 {\Large Guillaume Provencher}\thanks{\tt provench@dms.umontreal.ca} \\[2mm] \textit{Département de physique} \\ \textit{Université de Montréal, C.P. 6128, succ. centre-ville} \\ \textit{Montréal, QC, Canada, H3C 3J7}\\[2mm]
{\Large Yvan Saint-Aubin}\thanks{\tt saint@dms.umontreal.ca} \\[2mm] \textit{Département de mathématiques et de statistique} \\ \textit{Université de Montréal, C.P. 6128, succ. centre-ville} \\ \textit{Montréal, QC, Canada, H3C 3J7}
 }
 \date{\today}

\begin{document}
\maketitle  
\begin{abstract}
\noindent The vector space $\cn$ upon which the XXZ Hamilonian with $n$ spins acts bears the structure of a module over both the Temperley-Lieb algebra $\tl(\beta=q+q^{-1})$ and the quantum algebra $\uq$. The decomposition of $\cn$ as a $\uq$-module was first described by \citet{Rosso1988,Lusztig1988} and \citet{PS1990} and that as a $\tl$-module by \citet{Martin1992} (see also \citet{RS2007} and \citet{GV2012}). For $q$ generic, i.e.~not a root of unity, the $\tl$-module $\cn$ is known to be a sum of irreducible modules. We construct the projectors (idempotents of the algebra of endomorphisms of $\cn$) onto each of these irreducible modules as linear combinations of elements of $\uq$. When $q=q_c$ is a root of unity, the $\tl$-module $\cn$ (with $n$ large enough) can be written as a direct sum of indecomposable modules that are not all irreducible. We also give the idempotents projecting onto these indecomposable modules. Their expression now involves some new generators, whose action on $\cn$ is that of the divided powers $(S^\pm)^{(r)}=\lim_{q\rightarrow q_c} (S^\pm)^r/[r]!$. 

\end{abstract}

\tableofcontents

\newpage

%
\section{Introduction}
%

The XXZ Hamiltonian describing the dynamics of spin-$\frac12$ chains remains a crucial laboratory for theoretical physicists, mainly because its rich algebraic structure allows one to hope that the limit from finite lattices to the corresponding continuum theories can be fully understood. A common version of these models is defined by a Hamiltonian expressed as the sum of the generators of the Temperley-Lieb algebra $\tl(\beta)$. The anisotropy in the $z$-direction, as well as the boundary terms, are parametrized by a parameter $q\in \mathbb C^\times$ with the defining constant $\beta$ of $\tl$ being $\beta=q+q^{-1}$. When this parameter is a root of unity of the form $\exp(i\pi/p)$, $p>2$, it is agreed (\citet{PS1990,ABBBQ1987}) that this spin chain is related, in the continuum limit $n\to\infty$, to a conformal field theory of central charge 
\begin{equation*}
	c=1-\frac6{(p-1)p},
\end{equation*}
that characterizes the family of minimal models. By means of the Bethe ansatz, algebraic equations determining the eigenvalues of the XXZ model can be written (\citet{BVV1982,BA2001,Nepomechie2003}). These may be used to gain insight about its spectrum. However the explicit expression of the eigenvalues and the complete structure of the spectrum are difficult to describe for finite $n$. (Note that, for the anti-ferroelectric sector ($\beta<-1$), \citet{DFJMN1993} were able to successfully diagonalize the Hamiltonian of the chain in the limit $n\to\infty$.) Another problem related to the Hamiltonian is whether its spectrum is real as it is not hermitian for general $q$. Still, because of its link with minimal models, its spectrum should be real  at least for physically relevant values of $q$ and, indeed, numerical investigations (\citet{ABBBQ1987,ABB1987}) based upon Bethe ansatz relations support this claim. More recently, some progress has been made by Korff and Weston who in \cite{KW2007} introduce an inner product with respect to which the Hamiltonian at a root of unity is hermitian. Unfortunately, the inner product is restricted to a proper subspace of the representation space and might not be extendable in a way leading to a proof of the reality of the full spectrum of the Hamiltonian. 

The decomposition of $\cn$ as a $\tl$-module has been known since the early work of \citet{Martin1992} (see also \cite{GV2012}). When $q$ is generic, the Temperley-Lieb algebra is semisimple and $\cn$ is then a direct sum of irreducible modules. When $q$ is a root of unity, the Temperley-Lieb algebra $\tl$ is non-semisimple for $n$ large enough and then the decomposition of $\cn$ includes in general irreducible and indecomposable modules that are not irreducible. Still, to our knowledge, no simple way to construct these submodules in $\cn$ is known. A natural way to do so is to compute the primitive idempotents that project onto each irreducible or indecomposable submodule of this space. This is the goal of the present paper. (The objects that we shall construct are projectors $\cn\rightarrow\cn$ whose images are the indecomposable submodules. They are not idempotents of $\tl$ {\em per se} but rather elements of the algebra of endomorphisms $\mathrm{End}_{\tl}\cn$ that are projectors.) Of the many symmetries that the Hamiltonian enjoys, one makes it possible to obtain these idempotents: The quantum algebra $\uq$ and the duality existing between this algebra and $\tl$, known as the quantum Schur-Weyl duality.

The paper is organized as follows. First, we recall some definitions and then give a brief review of the quantum algebra $\uq$ and its representation theory, the Temperley-Lieb algebra and Schur-Weyl duality. The following two sections construct the idempotents, first, in the case when $q$ is generic and, second, when $q$ is a root of unity. Concluding remarks follow.

\section{Preliminaries}
This section introduces the two algebras $\tl$ and $\uq$ and recalls basic results. Standard notations are used throughout. The \emph{$q$-number} $[k]_q$ is
	\begin{equation}
		\label{eq:qnb}
		[k]_q := \frac{q^k-q^{-k}}{q-q^{-1}} = q^{k-1}+q^{k-3}\dotsb q^{-(k-3)}+q^{-(k-1)},
	\end{equation}
where $q\in\C^\times$. We shall write $[k]$ instead of $[k]_q$. The \emph{$q$-binomial coefficient} is
\begin{equation}
	\label{eq:qbin}
	\qbin{k}{l} := \frac{[k]!}{[l]!\,[k-l]!},
\end{equation}
where $[x]!=[x]\,[x-1]\dotsm[1]$ with $[0]!=1$. We set $[0]=0$, but $\qbinmini{k}{0}=1$ for $k\ge 0$. Note that $[k]\to k$ and $\qbinmini{k}{l}\to\binmini{k}{l}$ as $q\to1$. Like the standard binomial coefficient, the $q$-analog vanishes if $l>k$.

Roots of unity will be characterized by an integer $p$. This positive integer $p\ge 2$ is the smallest such that $q^{2p}=1$. We then say that the root of unity $q$ is associated with the integer $p$. In other words, a root of unity associated with $p$ is of the form $q=\exp\bigl(\frac{i\pi l}{p}\bigr)$, where $l$ and $p$ are coprime. An important point is that, if $q$ is such a root, then $[kp]=0$ for all $k\in\Z$. When $q$ is not root of unity, it is said to be \emph{generic}.
\subsection{The algebra $\uq$}\label{sect:algebraUq}

The algebra $\uq$, also known as the \emph{quantum algebra}, is a quasi-triangular Hopf algebra with unit $1_{U_q}$ generated by $\bigl\{S^\pm,q^{\pm S^z}\bigr\}$ under the relations (see, for example, \cite{ChariPressley,KS1997})
	\begin{equation}
		\label{eq:Uqsl2}
		q^{S^z}S^\pm q^{-S^z} = q^{\pm1}S^\pm,\qquad\bigl[S^+,S^-\bigr]=\bigl[2S^z\bigr],\qquad q^{S^z}q^{-S^z}=q^{-S^z}q^{S^z}=1_{U_q},
	\end{equation}
	where $q\in\C^\times$. The coproduct $\Delta:\uq\to\uq\,\otimes\,\uq$, the antipode $\gamma:\uq\to\uq$ and the counit $\varepsilon:\uq\to\C$ are defined respectively by
	\begin{equation*}
		\begin{gathered}\Delta\bigl(S^\pm\bigr) = q^{S^z}\otimes S^\pm+S^\pm\otimes q^{-S^z}, \\ \Delta\bigl(q^{\pm S^z}\big) = q^{\pm S^z}\otimes q^{\pm S^z},\end{gathered}\qquad
		\begin{gathered}\gamma\bigl(S^\pm\bigr)	 = -q^{\pm1}S^\pm, \\ \gamma\bigl(q^{\pm S^z}\bigr) = q^{\mp S^z},\end{gathered}\qquad
		\begin{gathered}\varepsilon\bigl(S^\pm\bigr) = 0, \\ \varepsilon\bigl(q^{\pm S^z}\bigr) = 1.\end{gathered}
	\end{equation*}
In the limit $q\to1$, the formul\ae{} \eqref{eq:Uqsl2} reduce to the defining relations of $\mathrm{U}(\mathfrak{sl}_2)$. When $q$ is generic, the center of $\uq$ is generated by the Casimir element (\citet{Jimbo1985})
\begin{equation*}
	S^2=S^-S^++\bigl[S^z+1/2\bigr]^2-[1/2]^2.
\end{equation*}

A natural representation $\pi:\uq\to\mathrm{End}\,\C^2$ is given in terms of Pauli matrices by
\begin{equation*}
	S^\pm\mapsto\sigma^\pm,\qquad q^{\pm S^z}\mapsto q^{\pm\sigma^z/2},
\end{equation*}
where $\sigma^z=\left(\begin{smallmatrix} 1 & 0 \\ 0 & 1 \end{smallmatrix}\right)$, $\sigma^+=\left(\begin{smallmatrix} 0 & 1 \\ 0 & 0 \end{smallmatrix}\right)$ and $\sigma^-=\left(\begin{smallmatrix} 0 & 0 \\ 1 & 0 \end{smallmatrix}\right)$. We label the element of the basis for $\C^2$ by the usual $\ket{+}$ and $\ket{-}$. A representation on $\cn$ is obtained through the representation $\pi$ by the recursive use of the coproduct. The generators can be explicitly written as
\begin{align}
	\begin{split}
	\label{eq:Uqrep}
	\pi_n\bigl(q^{S^z}\bigr)	& = q^{\sigma^z/2}\otimes\dotsm\otimes q^{\sigma^z/2},																									\\[1mm]
	\pi_n\bigl(S^\pm\bigr)	& = \sum_{i=1}^nq^{\sigma^z/2}\otimes\dotsm\otimes q^{\sigma^z/2}\otimes\sigma^\pm\otimes q^{-\sigma^z/2}\otimes\dotsm \otimes q^{-\sigma^z/2} = \sum_{i=1}^nS^\pm_i,
	\end{split}
\end{align}
where in the last equation the matrix $S_i^\pm$ corresponds to the term in the previous sum where $\sigma^\pm$ appears at position $i$. We also define $S^z=\frac12\sum_{1\le i\le n}\sigma^z_i$ with $\sigma^z_i=\mathds{1}\otimes\dotsb\otimes\mathds{1}\otimes\sigma^z\otimes\mathds{1}\otimes\dotsb\otimes\mathds{1}$ wherein $\sigma^z$ is at position $i$ and $\mathds{1}$ is the $2\times2$ identity matrix. This matrix is diagonal in the usual spin basis $\mathcal B=\bigl\{\, \ket{s_1s_2\dots s_n}\, |\, s_i\in\{+1,-1\}\bigr\}$ of $\cn$ and the tensor product decomposes naturally into a direct sum of its eigenspaces
\begin{equation}\label{eq:decompocn}
	\cn = \bigoplus_{m=-n/2}^{n/2}W_m,
\end{equation}
where $m$ is such that $S^z|_{W_m}= m\cdot \mathds{1}$. The eigenspaces {\em are not} $\uq$-submodules, but are $\tl$-submodules. For a given $n$ the set $J$ of non-negative eigenvalues of $S^z$ is $\{0,1,\dots, n/2\}$ if $n$ is even or $\{\frac12,\frac32,\dots,\frac n2\}$ if it is odd. We shall write the representation of the generators without the $\pi_n$: The context will make it clear whether we are speaking of the algebra elements or of the representation.

In what follows, a review of the representation theory for $\uq$ is given, for the cases when $q$ is generic and when $q$ is a root of unity (\citet{PS1990}). 
\subsubsection{Finite-dimensional representations for generic $q$}
The decomposition of $\cn$ as an $\uq$-module when $q$ is generic is obtained by the same approach used for its decomposition as an $\mathrm{U}(\mathfrak{sl}_2)$-module. The result is a  direct sum of irreducible modules $U_j$:
\begin{equation}
	\label{eq:UqDecomp}
	\cn\cong\bigoplus_{ j\in J }\Gamma^{(n)}_{j}\,U_j,
\end{equation}
where $\Gamma^{(n)}_{j} = \binmini{n}{n/2-j}-\binmini{n}{n/2-j-1}$ is the number of isomorphic copies of the module $U_j$, now to be defined.\footnote{Since $\dim U_j=2j+1$, the binomial identity $\sum_{k=0}^n\binmini{n}{k}=2^n$ leads to $\sum_j(2j+1)\,\Gamma^{(n)}_j=2^n=\dim\cn$.} The modules $U_j$ are of dimension $(2j+1)$ and a basis is labeled as in the theory of angular momentum: $\bigl\{\ket{j,j},\ket{j,j-1},\dotsc,\ket{j,-j}\bigr\}$ where the second label in each ket refers to the eigenvalue of $q^{S^z}$ (see \eqref{eq:actionSpm} below). The \emph{highest weight vector} $\ket{j,j}$ is annihilated by $S^+$ and is an eigenvector of $q^{S^z}$ with eigenvalue $q^j$, called the \emph{highest weight}. The other members of $U_j$ are obtained from $\ket{j,j}$ by the action of $S^-$:
\begin{equation}
	\label{eq:vectModule}
	\ket{j,j-k} = (S^-)^{(k)}\ket{j,j},\qquad 0\leq k\leq2j
\end{equation}
where $(S^-)^{(k)}=(S^-)^{k}/[k]!$ is the $k$-th \emph{divided power} of $S^-$. The action of the generators of $\uq$ on $U_j$ is given by
\begin{equation}
	\label{eq:actionSpm}
	\begin{gathered}
		q^{\pm S^z}\ket{j,m}=q^{\pm m}\ket{j,m},\qquad S^\pm\ket{j,m}=[j\pm m+1]\ket{j,m\pm1}\quad\text{with}	\\[1mm] S^+\ket{j,j}=S^-\ket{j,-j}=0.
	\end{gathered}
\end{equation}
The Casimir element $S^2$ is diagonal on those modules $U_j$:
\begin{equation}
	\label{eq:S2eigenvalues}
	S^2\ket{j,m} = \bigl([j+1/2]^2-[1/2]^2\bigr)\ket{j,m},\qquad \textrm{for all\ }m.
\end{equation}
If $j$ and $j'$ are distinct, the values of the Casimir on $U_j$ and $U_{j'}$ are distinct for $q$ generic. Indeed the equality $[j+\frac12]^2=[j'+\frac12]^2$ amounts to $\sin^2(j+\frac12)\theta=\sin^2(j'+\frac12)\theta$ if $q=e^{i\theta}$ with $\theta\in\mathbb C$. Then the equality may occur only when $\theta$ is real and a rational multiple of $\pi$, that is, only when $q$ is a root of unity.

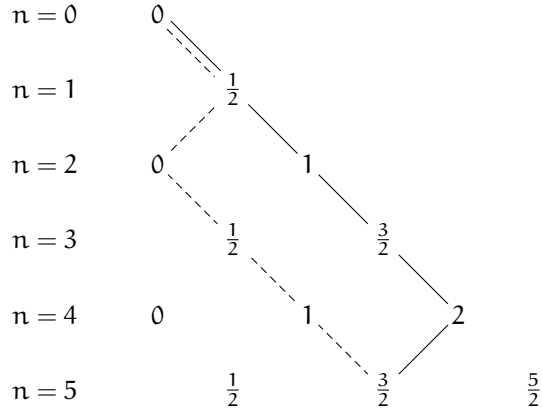
\begin{figure}[htb]
	\centering
	\begin{tikzpicture}[baseline={(current bounding box.center)},every node/.style={fill=white,circle,inner sep=1pt}]
	\draw (0,5.08) -- (1,4.08);\draw (1,4) -- (2,3);\draw (2,3) -- (3,2);\draw (3,2) -- (4,1);\draw (4,1) -- (3,0);
	\draw[densely dashed] (0,4.92) -- (1,3.92);\draw[densely dashed] (1,4) -- (0,3);\draw[densely dashed] (0,3) -- (1,2);\draw[densely dashed] (1,2) -- (2,1);\draw[densely dashed] (2,1) -- (3,0);
	\node at (0,5) {$0$};
	\node at (1,4) {$\frac12$};
	\node at (0,3) {$0$};\node at (2,3) {$1$};
	\node at (1,2) {$\frac12$};\node at (3,2) {$\frac32$};
	\node at (0,1) {$0$};\node at (2,1) {$1$};\node at (4,1) {$2$};
	\node at (1,0) {$\frac12$};\node at (3,0) {$\frac32$};\node at (5,0) {$\frac52$};
	\node at (-1.5,5) {$n=0$};
	\node at (-1.5,4) {$n=1$};
	\node at (-1.5,3) {$n=2$};
	\node at (-1.5,2) {$n=3$};
	\node at (-1.5,1) {$n=4$};
	\node at (-1.5,0) {$n=5$};
\end{tikzpicture}
	\caption[Bratteli diagram for $q$ generic]{\label{fig:BratDiag1}Bratteli diagram for $n=5$ with two sample paths.}
\end{figure}
Some information about the decomposition \eqref{eq:UqDecomp} is encoded in a \emph{Bratteli diagram}.
An example of such a diagram is shown on Figure \ref{fig:BratDiag1} for $n=5$. The values of $j$ labeling the admissible modules in the decomposition are shown on the $n$-th row, starting from the top one ($0$-th row). The number of downward paths starting from $(0,0)$ and reaching the pair $(n,j)$ is precisely $\Gamma^{(n)}_j$. For instance, there are $\Gamma^{(5)}_{3/2}=4$ such paths ending at $(n,j)=(5,3/2)$. Two are drawn on Figure \ref{fig:BratDiag1}.

%
\subsubsection{Finite-dimensional representations for $q$ a root of unity}\label{subsect:ru}
%

When $q$ is a root of unity, several of the previous observations used to describe the representation theory of $\uq$ fail. To describe them, let $q$ be a root of unity associated with $p$.

A first difference is that the Casimir no longer distinguishes between modules. Indeed, its values coincide on modules $U_j$ and $U_{j'}$ whenever $j$ and $j'$  are related by either
\begin{equation}
	\label{eq:orbit}
	j'\equiv j\modulo p\qquad \text{or}\qquad j'\equiv p-j-1\modulo p.
\end{equation}
It is useful to partition the set $J$ into orbits. If $j$ satisfies $2j+1\equiv 0\text{\, mod\,}p$, it is called {\em critical} and its orbit $\text{orb}_j$ is simply $\{j\}$. For any other $j$, the orbit $\text{orb}_j$ includes all elements of $J$ that are related to $j$ by either one of relations \eqref{eq:orbit}. Orbits can be read easily from the Bratteli diagram. First draw {\em critical lines}, that is vertical lines through the critical $j$'s. Then read the orbit of a given non-critical $j$ of the $n$-th row as the set of $j'$ obtained from $j$ by (possibly multiple) mirror reflections through the critical lines. The non-critical orbits for $n=20$ and $p=5$ appear on Figure \ref{fig:BratDiag2} where only the rows $19$ and $20$ are shown. The two orbits are represented by solid arcs ($\text{orb}_0$) and dotted ones ($\text{orb}_1$). We shall often use the leftmost element to label any given orbit.
\begin{figure}[htb]
	\centering
	\begin{tikzpicture}[baseline={(current bounding box.center)},every node/.style={fill=white,circle,inner sep=2pt},scale=2/3]
	\begin{scope}[thick]
		\draw (8,0) arc [start angle=0, end angle=-180, x radius=4, y radius=1.5] -- (0,0);
		\draw (10,0) arc [start angle=0, end angle=-180, x radius=1, y radius=1] -- (8,0);
		\draw (18,0) arc [start angle=0, end angle=-180, x radius=4, y radius=1.5] -- (10,0);
		\draw (20,0) arc [start angle=0, end angle=-180, x radius=1, y radius=1] -- (18,0);
	\end{scope}
	\begin{scope}[thick,densely dotted]
		\draw (6,0) arc [start angle=0, end angle=-180, x radius=2, y radius=1] -- (2,0);
		\draw (12,0) arc [start angle=0, end angle=-180, x radius=3, y radius=1.5] -- (6,0);
		\draw (16,0) arc [start angle=0, end angle=-180, x radius=2, y radius=1] -- (12,0);
	\end{scope}
	\draw[densely dashed] (4,-0.75) -- (4,2.25);
	\draw[densely dashed] (9,-0.75) -- (9,2.25);
	\draw[densely dashed] (14,-0.75) -- (14,2.25);
	\draw[densely dashed] (19,-0.75) -- (19,2.25);
	\node at (1,1) {$\frac12$};
	\node at (3,1) {$\frac32$};
	\node at (5,1) {$\frac52$};
	\node at (7,1) {$\frac72$};
	\node at (9,1) {$\frac92$};
	\node at (11,1) {$\frac{11}2$};
	\node at (13,1) {$\frac{13}2$};
	\node at (15,1) {$\frac{15}2$};
	\node at (17,1) {$\frac{17}2$};
	\node at (19,1) {$\frac{19}2$};
	\node at (0,0) {$0$};
	\node at (2,0) {$1$};
	\node at (4,0) {$2$};
	\node at (6,0) {$3$};
	\node at (8,0) {$4$};
	\node at (10,0) {$5$};
	\node at (12,0) {$6$};
	\node at (14,0) {$7$};
	\node at (16,0) {$8$};
	\node at (18,0) {$9$};
	\node at (20,0) {$10$};
\end{tikzpicture}
	\caption[Bratteli diagram for $q$ root of unity]{\label{fig:BratDiag2}Last two rows of the Bratteli diagram for $n=20$ and $p=5$, showing two orbits.}
\end{figure}
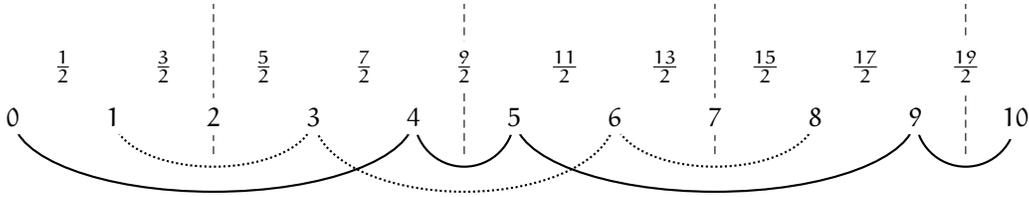

The second difference is that $(S^\pm)^p$ acts as zero on $\cn$. To see this, first observe that
\begin{equation*}
	S^\pm_jS^\pm_i = q^{\pm2}\,S^\pm_iS^\pm_j,\qquad i<j,
\end{equation*}
with $S^\pm_i$ as in \eqref{eq:Uqrep}. Then, using this relation, a computation shows that
\begin{equation}\label{eq:sp}
	\bigl(S^\pm\bigr)^p = q^{\pm p(p-1)/2}\,[p]!\sum_{i_1<\dotsb<i_p}S^\pm_{i_1}\dotsm S^\pm_{i_p},
\end{equation}
which is zero when $q=q_c$ is a root associated with $p$. The {\em divided power} $(S^\pm)^{(r)}$ defined as $\lim_{q\rightarrow q_c}(S^\pm)^r/[r]!$ is however well-defined for all $r$ and $(S^\pm)^{(p)}$ is nonzero on $\cn$ if $p\leq n$. There is not necessarily an element in $\uq$ whose action on $\cn$ coincides with that of $(S^\pm)^{(p)}$. One may extend $\uq$ into a larger algebra including the elements $(S^\pm)^{(p)}$ and whose defining relations at $q=q_c$ are obtained as limits of those at generic $q$. (See, for example, \citet{Martin1992, BFGT2009}, and \citet{GV2012}, where this is done in a context close to physical applications. Note that the exact definition of the divided powers differs slightly between authors.) The resulting algebra is often called the Lusztig extension of $\uq$. The new elements satisfy
\begin{equation}
	\label{eq:SpmRenorm}
	\Bigl[\bigl(S^+\bigr)^{(k)},\bigl(S^-\bigr)^{(l)}\Bigr] = \sum_{i=1}^k\qbin{2S^z+l-k}{i}\bigl(S^-\bigr)^{(l-i)}\bigl(S^+\bigr)^{(k-i)},
\end{equation}
for $l\geq k\geq0$ (see \cite{KS1997}).

The decomposition of $\cn$ as a direct sum of modules is done orbit by orbit, as follows. Given $\mathrm{orb}_{j_1}=\{j_1,\dotsc,j_k\}$ with $j_1<\dotsb<j_k$, choose a vector $v\in W_{j_k}\subset\cn$, that is ${S^z}v={j_k}v$, such that $S^+v=0$ and write $v$ as $\ket{j_k,j_k}$. This is a highest weight vector and it generates a module under the action of the generators (and the divided powers $(S^\pm)^{(p)}$). Its ``descendants'' are defined by the relation \eqref{eq:vectModule} and the action of $S^\pm$ on this module is given by \eqref{eq:actionSpm}, as in the generic case. Here however, some vectors become unreachable by the action of $S^\pm$. For instance, if $j_c=\frac12(rp-1)$ is the first critical line to the left of $j_k$ on the Bratteli diagram and $j_{k-1} = 2j_c-j_k$ is non negative, then $j_{k-1}\in \mathrm{orb}_{j_1}$ and
\begin{align*}
	S^+\ket{j_k,j_{k-1}}	& = [j_k+j_{k-1}+1]\ket{j_k,j_{k-1}+1}	\\
					& = [rp]\ket{j_k,j_{k-1}+1}				\\
					& = 0.
\end{align*}
Of course, the same happens for any vector $\ket{j_k,j_{k-1}-lp}$ with $l\geq0$. A similar situation occurs with $S^-$ on $\ket{j_k,j_k-lp+1}$. 
\begin{figure}[htb]
	\centering
	\begin{tikzpicture}[baseline={(current bounding box.center)},>=stealth,thick]
	\begin{scope}[thick]
		\draw (2,0) -- (1.25,0) node[left=5pt] {$\scriptstyle j_{k-3}\,=\,j_{k-1}-p$};
		\draw (2,0.5) -- (1.25,0.5) node[left=5pt] {$\scriptstyle j_{k-1}-p+1$};
		\draw (2,1.75) -- (1.25,1.75) node[left=5pt] {$\scriptstyle j_{k-2}\,=\,j_k-p$};
		\draw (2,2.25) -- (1.25,2.25) node[left=5pt] {$\scriptstyle j_{k}-p+1$};
		\draw (2,3.25) -- (1.25,3.25) node[left=5pt] {$\scriptstyle j_{k-1}$}; 
		\draw (2,3.75) -- (1.25,3.75) node[left=5pt] {$\scriptstyle j_{k-1}+1$}; 
		\draw (2,5) -- (1.625,5) node[above=5pt] {$\scriptstyle\ket{j_k,\,j_k}$};\draw (1.625,5) -- (1.25,5) node[left=5pt] {$\scriptstyle j_k$};
	\end{scope}
	\begin{scope}[thin]
		\draw[densely dotted] (1.5,0) -- (1.5,-0.5);\draw[densely dotted] (1.75,0) -- (1.75,-0.5);
		\draw[densely dotted] (1.5,0.5) -- (1.5,1.75);\draw[densely dotted] (1.75,0.5) -- (1.75,1.75);
		\draw[densely dotted] (1.5,2.25) -- (1.5,3.25);\draw[densely dotted] (1.75,2.25) -- (1.75,3.25);
		\draw[densely dotted] (1.5,3.75) -- (1.5,5);\draw[densely dotted] (1.75,3.75) -- (1.75,5);
	\end{scope}
	\node at (1.625,-0.75) {\large\vdots};
	\begin{scope}[thin]
		\draw[->] (1.5,0.5) -- (1.5,0);
		\draw[->] (1.75,1.75) -- (1.75,2.25);
		\draw[->] (1.5,3.75) -- (1.5,3.25);
	\end{scope}
	\begin{scope}[thick]
		\draw (3.5,0) -- (4.25,0);\draw (3.5,0.5) -- (4.25,0.5);
		\draw (3.5,1.75) -- (4.25,1.75);\draw (3.5,2.25) -- (4.25,2.25);
		\draw (3.5,3.25) -- (3.875,3.25) node[above=5pt]{$\scriptstyle\ket{j_{k-1},\,j_{k-1}}$};\draw (3.875,3.25) -- (4.25,3.25);
	\end{scope}
	\begin{scope}[thin]
		\draw[densely dotted] (3.75,0) -- (3.75,-0.5);\draw[densely dotted] (4,0) -- (4,-0.5);
		\draw[densely dotted] (3.75,0.5) -- (3.75,1.75);\draw[densely dotted] (4,0.5) -- (4,1.75);
		\draw[densely dotted] (3.75,2.25) -- (3.75,3.25);\draw[densely dotted] (4,2.25) -- (4,3.25);
	\end{scope}
	\node at (3.875,-0.75) {\large\vdots};
	\begin{scope}[thin]
		\draw[->] (4,0) -- (4,0.5);
		\draw[->] (3.75,2.25) -- (3.75,1.75);
	\end{scope}
	\begin{scope}[thin]
		\draw[<-] (2+0.1423,3.75-0.05) -- (3.5-0.1423,3.25+0.05);
		\draw[<-] (2+0.1423,2.25-0.05) -- (3.5-0.1423,1.75+0.05);
		\draw[<-] (2+0.1423,0.5-0.05) -- (3.5-0.1423,0+0.05);
	\end{scope}
\end{tikzpicture}
	\caption[Paired module $U_{j_k,j_{k-1}}$]{\label{fig:pairing}Tower illustration of the paired module $U_{j_k,j_{k-1}}$. An up or down arrow means respectively that the action of $S^+$ or $S^-$ is non-vanishing, while the dotted lines mean that both actions are non-vanishing.}
\end{figure}
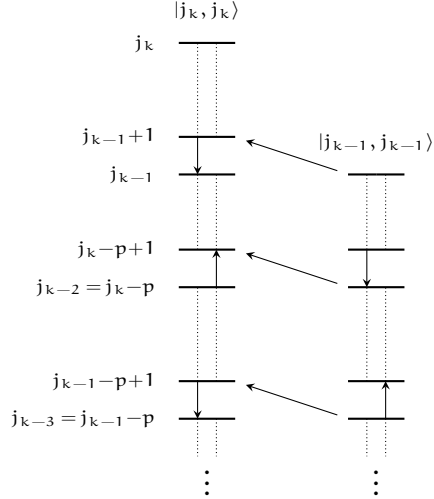

Let $j_k$ and $j_{k-1}$ as above. For $q$ generic, the eigenspace $W_{j_{k-1}}$ of $S^z$ contains a subspace belonging to the direct sum of irreducibles $U_j$ with $j>j_{k-1}$. A complement in $W_{j_{k-1}}$ may be chosen to coincide with $\ker S^+|_{W_{k-1}}$ and its dimension is $\Gamma^{(n)}_{j_{k-1}}$. A basis $\{ \ket{j_{k-1},j_{k-1}}_i, i=1, \dots, \Gamma^{(n)}_{j_{k-1}}\}$ is then constituted of highest weight vectors. For $q$ a root of unity, one can show that, for each highest weight vector $\ket{j_k,j_k}$, there exists a vector $w\in W_{j_{k-1}}$ such that $S^+w$ is non-zero and equal to $\ket{j_k,j_{k-1}+1}$. We shall write it as $\ket{j_{k-1},j_{k-1}}$. From this $w=\ket{j_{k-1},j_{k-1}}$ a subspace is generated by the action of the divided powers of $S^-$
\begin{equation*}
	\ket{j_{k-1},j_{k-1}-r} = (S^-)^{(r)}\ket{j_{k-1},j_{k-1}}.
\end{equation*}
Together with the descendants of $\ket{j_k,j_k}$, they span a $\uq$-submodule of dimension $2(j_k+j_{k-1}+1)$ where the action is given by \eqref{eq:actionSpm} supplemented by the
relations
\begin{equation}
	\label{eq:SpmSmallTower}
	\begin{split}
		S^-\ket{j_{k-1},m} 		& = [j_{k-1}-m+1]\ket{j_{k-1},m-1}									\\[2mm]
		S^+\ket{j_{k-1},m} 		& = [j_{k-1}+m+1]\ket{j_{k-1},m+1} + \qbin{j_k-m-1}{j_{k-1}-m}\ket{j_k,m+1}
	\end{split}
\end{equation}
valid for $j_{k-1}\geq m> -j_{k-1}$ for the first and $j_{k-1}> m\geq -j_{k-1}$ for the second, and
\begin{align*}
	S^+\ket{j_{k-1},j_{k-1}} 	& = \ket{j_k,j_{k-1}+1}												\\[2mm]
	S^-\ket{j_{k-1},-j_{k-1}} 	& = \bin{j_k+j_{k-1}-1}{2\,j_{k-1}}\frac{[j_k+j_{k-1}]}{[j_k-j_
{k-1}]}\ket{j_k,-j_{k-1}-1}.
\end{align*}
(These first appeared in \cite{PS1990}.) This $\uq$-module will be denoted by $U_{j_k,j_{k-1}}$. (It is projective as a module over the extended algebra described earlier.) Its structure is depicted on Figure \ref{fig:pairing}.

The above procedure can be repeated until every vector in a basis of $\ker S^+|_{W_{j_k}}$ is paired to one in $W_{j_{k-1}}$. The procedure described for the vector $\ket{j_k,j_k}$ is then repeated for all highest weight vectors of weight $j_{k-1}$ (that is, vectors in a basis for $\ker S^+|_{W_{j_{k-1}}}$), matching each with a partner in $W_{j_{k-2}}$. This pairing constructs 
\begin{equation*}
	\Omega^{(n)}_{j_i} = \Gamma^{(n)}_{j_i}-\Gamma^{(n)}_{j_{i+1}}+\Gamma^{(n)}_{j_{i+2}}-\Gamma^{(n)}_{j_{i+3}}+\dotsb+(-1)^{k-i}\Gamma^{(n)}_{j_k}
\end{equation*}
modules isomorphic to $U_{j_i,j_{i-1}}$ for $i\ge 2$. If $\ker S^+|_{W_{j_0}}$ is non-zero, that is if $\Omega^{(n)}_{j_1}>0$, then the highest weight vectors in $W_{j_0}$ cannot be paired as there is no $j_0$ in the orbit. They generate, by the action of the divided powers $(S^-)^{(r)}$, modules whose structure is similar to the module $U_j$ appearing in the generic case in the sense that $S^+$ (resp.~$S^-$) vanishes only on the highest weight vector (resp.~lowest one). They are irreducible and will also be labeled by $U_j$. The procedure is repeated for each non-critical orbit.

For critical $j=j_c$, all the highest weight vectors of weight $j_c$ lead to modules $M_{j_c}$, with action prescribed by \eqref{eq:actionSpm}. These are irreducible as modules over the extended algebra, but not necessarily over $\uq$ as the action of $S^-$ on vectors $\ket{j_c,j_c-lp+1}$ vanishes for all $l$, as well as that of $S^+$ on $\ket{j_c,-j_c+lp-1}$. Contrarily to the structure of the $U_{j_i,j_{i-1}}$ modules with $i\ge 2$ depicted in Figure \ref{fig:pairing}, their graphical representation is made of a single tower as in the generic case.

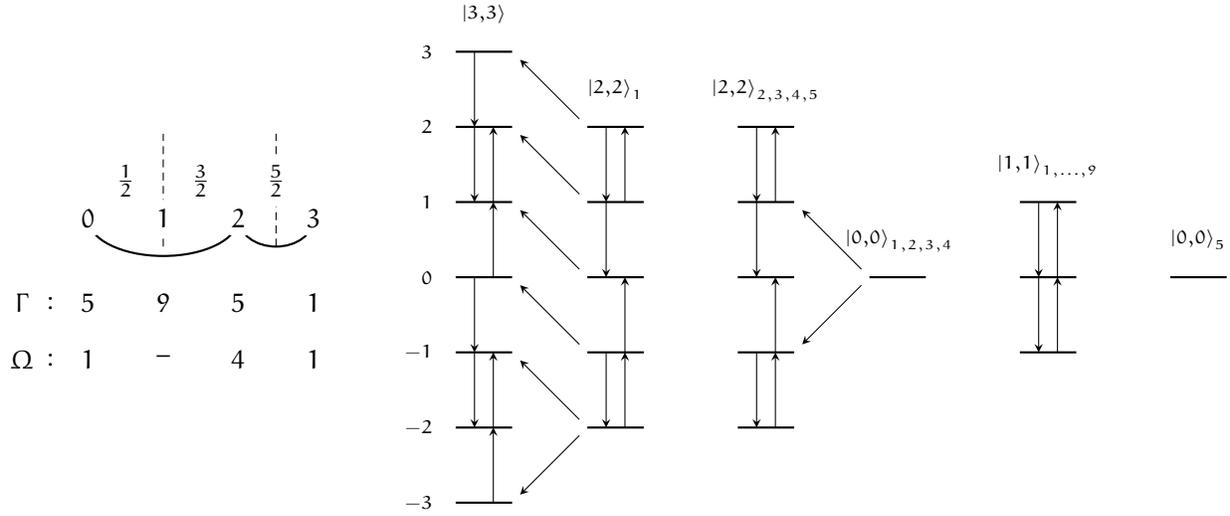
\begin{figure}[htb]
	\centering
	\begin{tikzpicture}[baseline={(current bounding box.center)},every node/.style={fill=white,circle,inner sep=2pt},scale=1/2]
	\begin{scope}[thick]
		\draw (4,0) arc [start angle=0, end angle=-180, x radius=2, y radius=1] -- (0,0);\draw (6,0) arc [start angle=0, end angle=-180, x radius=1, y radius=0.75] -- (4,0);
	\end{scope}
	\draw[densely dashed] (2,-0.75) -- (2,2.25);\draw[densely dashed] (5,-0.75) -- (5,2.25);
	\node at (1,1) {$\frac12$};\node at (3,1) {$\frac32$};\node at (5,1) {$\frac52$};
	\node at (0,0) {$0$};\node at (2,0) {$1$};\node at (4,0) {$2$};\node at (6,0) {$3$};
	\node at (-1.75,-2.25) {$\Gamma$};\node at (-1,-2.25) {$:$};\node at (0,-2.25) {$5$};\node at (2,-2.25) {$9$};\node at (4,-2.25) {$5$};\node at (6,-2.25) {$1$};
	\node at (-1.75,-3.75) {$\Omega$};\node at (-1,-3.75) {$:$};\node at (0,-3.75) {$1$};\node at (2,-3.75) {--};\node at (4,-3.75) {$4$};\node at (6,-3.75) {$1$};
\end{tikzpicture}\hfill
	\begin{tikzpicture}[baseline={(current bounding box.center)},>=stealth,thick]
%
%
	\begin{scope}[thick]
		\foreach \y in {-3,-2,...,3}
			\draw (2,\y) -- (1.25,\y) node[left=5pt] {$\scriptstyle\y$};
		\node at (1.625,3.5) {$\scriptstyle\ket{3,3}$};
	\end{scope}
	\begin{scope}[thick]
		\foreach \y in {-2,-1,...,2}
			\draw (3.75,\y) -- (3,\y);
		\node at (3.375,2.5) {$\scriptstyle\ket{2,2}_1$};
	\end{scope}[thick]
	\begin{scope}[thin]
		\draw[->] (1.5,3) -- (1.5,2);
		\draw[->] (1.5,2) -- (1.5,1);\draw[<-] (1.75,2) -- (1.75,1);
		\draw[<-] (1.75,1) -- (1.75,0);
		\draw[->] (1.5,0) -- (1.5,-1);
		\draw[->] (1.5,-1) -- (1.5,-2);\draw[<-] (1.75,-1) -- (1.75,-2);
		\draw[<-] (1.75,-2) -- (1.75,-3);
	\end{scope}
	\begin{scope}[thin]
		\draw[->] (3.25,2) -- (3.25,1);\draw[<-] (3.5,2) -- (3.5,1);
		\draw[->] (3.25,1) -- (3.25,0);
		\draw[<-] (3.5,0) -- (3.5,-1);
		\draw[->] (3.25,-1) -- (3.25,-2);\draw[<-] (3.5,-1) -- (3.5,-2);
	\end{scope}
	\begin{scope}[thin]
		\foreach \y in {3,2,...,-1}
			\draw[<-] (2+0.10606,\y-0.10606) -- (3-0.10606,\y-1+0.10606);
		\draw[<-] (2+0.10606,-3+0.10606) -- (3-0.10606,-2-0.10606);
	\end{scope}
%
	\begin{scope}[thick]
		\foreach \y in {-2,-1,...,2}
			\draw (5.75,\y) -- (5,\y);
		\node at (5.375,2.5) {$\scriptstyle\ket{2,2}_{2,3,4,5}$};
	\end{scope}
	\begin{scope}[thick]
		\draw (7.5,0) -- (6.75,0);
		\node at (7.15,0.5) {$\scriptstyle\ket{0,0}_{1,2,3,4}$};
	\end{scope}
	\begin{scope}[thin]
		\draw[->] (5.25,2) -- (5.25,1);\draw[<-] (5.5,2) -- (5.5,1);
		\draw[->] (5.25,1) -- (5.25,0);
		\draw[<-] (5.5,0) -- (5.5,-1);
		\draw[->] (5.25,-1) -- (5.25,-2);\draw[<-] (5.5,-1) -- (5.5,-2);
	\end{scope}
	\begin{scope}[thin]
		\draw[<-] (5.75+0.10606,1-0.10606) -- (6.75-0.10606,0+0.10606);
		\draw[<-] (5.75+0.10606,-1+0.10606) -- (6.75-0.10606,0-0.10606);
	\end{scope}
%
%
	\begin{scope}[thick]
		\foreach \y in {1,0,-1}
			\draw (9.5,\y) -- (8.75,\y);
	\end{scope}

	\node at (9.125,1.5) {$\scriptstyle\ket{1,1}_{1,\dotsc,9}$};

	\begin{scope}[thin]
		\draw[->] (9,1) -- (9,0);\draw[<-] (9.25,1) -- (9.25,0);
		\draw[->] (9,0) -- (9,-1);\draw[<-] (9.25,0) -- (9.25,-1);
	\end{scope}

	\draw[thick] (11.5,0) -- (10.75,0);\node at (11.125,0.5) {$\scriptstyle\ket{0,0}_5$};
\end{tikzpicture}
	\caption[Bratteli diagram and towers for $n=6$, $p=3$]{\label{fig:n6p3}Last rows of the Bratteli diagram and towers for $n=6$, $p=3$ ($q^6=1$).}
\end{figure}
The decomposition of $\cn$ can therefore be written as the following direct sum:
\begin{equation*}
	\cn\cong\Big(\bigoplus_{\substack{ j_i\,\in\,\mathrm{orb}_j, i\ge 2 \\[1pt] j\,<\,\frac12(p-1)}}\Omega^{(n)}_{j_i}\,U_{j_i,j_{i-1}}\Big)\oplus
	\Big(\bigoplus_{j\,<\,\frac12(p-1)}\Omega^{(n)}_{j}\,U_{j}\Big)\oplus 
	\Big(\bigoplus_{j_c\,\text{critical}}\Gamma^{(n)}_{j_c}\,M_{j_c}\Big).
\end{equation*}
The first sum includes all modules $U_{j_i,j_{i-1}}$ constructed by the pairing procedure, the second the $U_{j}$ corresponding to the first element $j_1$ of each orbit and the last one the $M_{j_c}$ associated with the critical $j_c$. As an example, the decomposition of $\cn$ for $n=6$, $p=3$ ($q^6=1$) is given. The fifth and sixth rows of the Bratteli diagram are shown on Figure \ref{fig:n6p3} as well as the paired ``towers'' resulting from the decomposition
\begin{equation*}
	\otimes^6\C^2\cong \left(1\cdot U_{3,2}\oplus 4\cdot U_{2,0}\right)\oplus \left( 1\cdot U_0\right )\oplus\left(9\cdot M_1\right).
\end{equation*}
Of course the dimensions of the indecomposables ($(1\cdot 12 +4\cdot 6)+(1\cdot 1)+(9\cdot 3)$) sum up correctly to $2^6=64$.
\subsection{The algebra $\tl$}\label{sect:algebraTL}

\subsubsection{$\tl$ and its representation theory for generic $q$}\label{sec:ton}

This section gathers basic results about the Temperley-Lieb algebra and its representation theory.
\begin{defn}[Temperley-Lieb algebra]
	\label{defn:TLn}
	For $q\in\C^\times$ and $n\ge 1$, the unital associative algebra over $\C$ generated by the elements $\{e_1,\dotsc,e_{n-1}\}$ satisfying
	\begin{align*}
		e_i^2			& = \beta e_i,\qquad\beta=[2]=q+q^{-1},		\\
		e_i\,e_{i\pm1}\,e_i 	& = e_i, 								\\
		e_i\,e_j			& = e_j\,e_i,\qquad|i-j| > 1,
	\end{align*}
	is called the \emph{Temperley-Lieb algebra} $\tl(q)$. The case $n=1$ is $\tlsans_1(q)=\C$.
\end{defn}

The representation theory of $\tl$ for general $q$ has been known since the early work of \citet{GoodmanWenzl1993} and \citet{Martin1991}. (See also \cite{GV2012, RSA2012}.) For $q$ generic, their fundamental result is that $\tl$ is a semisimple algebra and, therefore, all its (finite-dimensional) modules are direct sums of irreducible ones. A complete set of non-isomorphic irreducible modules is constituted of modules $\mathcal V_{n,m}$ with $0\le m\le n/2$ and $n/2-m\in\mathbb N$ of dimension $\Gamma^{(n)}_m=\left(\begin{smallmatrix} n\\ n/2-m\end{smallmatrix}\right) - \left(\begin{smallmatrix} n\\ n/2-m-1\end{smallmatrix}\right)$ and, of course,
$$\dim\tl=\frac1{n+1}\begin{pmatrix} 2n\\ n\end{pmatrix}=\sum_{0\le m\le n/2}(\dim \mathcal V_{n,m})^2.$$
When $q$ is generic, the irreducible modules $\mathcal V_{n,m}$ (or simply $\mathcal V_m$) have a natural description in terms of {\em standard modules}. These standard modules are defined for all $q$ and we now recall their graphical description in terms of connectivities and link states. The connectivities are rectangles with $n$ points on each of their left and right sides, all these points being connected pairwise by non-intersecting curves drawn within the rectangles. The generators $e_i$ correspond to 
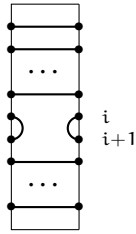
\begin{center}
	\begin{tikzpicture}[baseline={(current bounding box.center)},scale=0.3]
    \draw (3,0) -- (0,0) -- (0,10) -- (3,10);
    \draw (3,10) -- ( 3,5) node[right=5pt] {$\scriptstyle i$};
    \draw (3,5) -- (3,4) node[right=5pt] {$\scriptstyle i+1$};
    \draw (3,4) -- (3,0);
    \foreach \y in {1,3,4,5,6,8,9}
    	\node at (0,\y) {$\scriptstyle\bullet$};
    \foreach \y in {1,3,4,5,6,8,9}
    	\node at (3,\y) {$\scriptstyle\bullet$};


    \begin{scope}[thick]
    	\draw (0,1) -- (3,1);
    	\draw (0,3) -- (3,3);
    	\draw (0,6) -- (3,6);
    	\draw (0,8) -- (3,8);
    	\draw (0,9) -- (3,9);
	\draw (0,4) arc (270:450:0.5);
	\draw (3,5) arc (90:270:0.5);
    \end{scope}
    \node at (1.5,2) {$\dots$};
    \node at (1.5,7) {$\dots$};
\end{tikzpicture}
\end{center}
and the multiplication in connectivities can be done by juxtaposing the two rectangles, reading how the remaining $2n$ points on the left and right are connected, and multiplying the resulting connectivity by a factor $\beta=q+q^{-1}$ for each closed loop. Here is an example of a product of two elements of $\mathrm{TL}_4$:
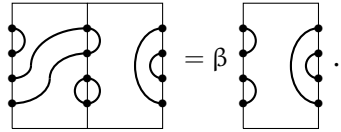
\begin{equation*}
	\begin{tikzpicture}[baseline={(current bounding box.center)},scale=1/3]
	\draw (0,0) -- (3,0) -- (3,5) -- (0,5) -- (0,0);
 	\draw (3,0) -- (6,0) -- (6,5) -- (3,5);
	\foreach \y in {1,2,3,4}
    		\node at (0,\y) {$\scriptstyle\bullet$};
	\foreach \y in {1,2,3,4}
		\node at (3,\y) {$\scriptstyle\bullet$};
	\foreach \y in {1,2,3,4}
    		\node at (6,\y) {$\scriptstyle\bullet$};
	\begin{scope}[thick]
		\draw (0,3) arc [start angle=-90, end angle=90, x radius=0.5, y radius=0.5];
		\draw (3,1) arc [start angle=-90, end angle=270, x radius=0.5, y radius=0.5];
		\draw (3,3) arc [start angle=-90, end angle=90, x radius=0.5, y radius=0.5];
		\draw (3,4) arc [start angle=90, end angle=180, x radius=2.25, y radius=1.5];\draw (0.75,2.5) arc [start angle=0, end angle=-90, x radius=0.75, y radius=0.5];
		\draw (3,3) arc [start angle=90, end angle=180, x radius=1.5, y radius=1];\draw (1.5,2) arc [start angle=0, end angle=-90, x radius=1.5, y radius=1];
		\draw (6,2) arc [start angle=270, end angle=90, x radius=0.5, y radius=0.5];
		\draw (6,1) arc [start angle=270, end angle=90, x radius=1.1, y radius=1.5];
	\end{scope}
\end{tikzpicture} = \beta\begin{tikzpicture}[baseline={(current bounding box.center)},scale=1/3]
	\draw (0,0) -- (3,0) -- (3,5) -- (0,5) -- (0,0);
	\foreach \y in {1,2,3,4}
    		\node at (0,\y) {$\scriptstyle\bullet$};
	\foreach \y in {1,2,3,4}
    		\node at (3,\y) {$\scriptstyle\bullet$};
	\begin{scope}[thick]
		\draw (0,1) arc [start angle=-90, end angle=90, x radius=0.5, y radius=0.5];
		\draw (0,3) arc [start angle=-90, end angle=90, x radius=0.5, y radius=0.5];
		\draw (3,2) arc [start angle=270, end angle=90, x radius=0.5, y radius=0.5];
		\draw (3,1) arc [start angle=270, end angle=90, x radius=1.1, y radius=1.5];
	\end{scope}
\end{tikzpicture}.
\end{equation*}
Words in $\tl$, that is products of generators, are in one-to-one correspondence with connectivities and general elements in $\tl$ are linear combinations of them. The faithfulness of this graphical description is shown, for example, in \cite{RSA2012}. The standard module $\mathsf V_{n,\ell}$ is described by giving a basis and the action of connectivities on this basis. A basis for $\mathsf V_{n,\ell}$ is the set of link vectors whose graphical description is given by a straight vertical segment with $n$ dots, $2\ell$ of which are tied pairwise by non-intersecting curves drawn to the right of the segment. The remaining $n-2\ell$ points are indicated by horizontal segments. The bases for $\mathsf V_{4,2}$, $\mathsf V_{4,1}$ and $\mathsf V_{4,0}$ are
$$\mathcal B_{4,2}=\left\{ \begin{tikzpicture}[scale=0.2]
\useasboundingbox (-0.5,2) rectangle (2,4);
    \draw[line width=0.5pt] (0,0) -- (0,5);
    	\draw[line width=0.5pt] (0,1) arc (270:450:1.5);
    	\draw[line width=0.5pt] (0,2) arc (270:450:0.5);
\end{tikzpicture},\begin{tikzpicture}[scale=0.2]
\useasboundingbox (-0.5,2) rectangle (1,4);
    \draw[line width=0.5pt] (0,0) -- (0,5);
    	\draw[line width=0.5pt] (0,1) arc (270:450:0.5);
    	\draw[line width=0.5pt] (0,3) arc (270:450:0.5);
\end{tikzpicture}\right\},\qquad
\mathcal B_{4,1}=\left\{ \begin{tikzpicture}[scale=0.2]
\useasboundingbox (-0.5,2) rectangle (1.5,4);
    \draw[line width=0.5pt] (0,0) -- (0,5);
    	\draw[line width=0.5pt] (0,1) -- (0.75,1);
    	\draw[line width=0.5pt] (0,2) -- (0.75,2);
    	\draw[line width=0.5pt] (0,3) arc (270:450:0.5);
\end{tikzpicture},\begin{tikzpicture}[scale=0.2]
\useasboundingbox (-0.5,2) rectangle (1.5,4);
    \draw[line width=0.5pt] (0,0) -- (0,5);
    	\draw[line width=0.5pt] (0,1) -- (0.75,1);
    	\draw[line width=0.5pt] (0,4) -- (0.75,4);
    	\draw[line width=0.5pt] (0,2) arc (270:450:0.5);
\end{tikzpicture},\begin{tikzpicture}[scale=0.2]
\useasboundingbox (-0.5,2) rectangle (1.5,4);
    \draw[line width=0.5pt] (0,0) -- (0,5);
    	\draw[line width=0.5pt] (0,3) -- (0.75,3);
    	\draw[line width=0.5pt] (0,4) -- (0.75,4);
    	\draw[line width=0.5pt] (0,1) arc (270:450:0.5);
\end{tikzpicture}\right\},\qquad
\mathcal B_{4,0}=\left\{ \begin{tikzpicture}[scale=0.2]
\useasboundingbox (-0.5,2) rectangle (1.5,4);
    \draw[line width=0.5pt] (0,0) -- (0,5);
    	\draw[line width=0.5pt] (0,1) -- (0.75,1);
    	\draw[line width=0.5pt] (0,2) -- (0.75,2);
    	\draw[line width=0.5pt] (0,3) -- (0.75,3);
    	\draw[line width=0.5pt] (0,4) -- (0.75,4);
\end{tikzpicture}\right\}.$$
The action of $\tl$ on $\mathsf V_{n,\ell}$ is defined graphically as for the product in $\tl$ in its graphical description with the additional rule that, if the resulting link vector has more than $2\ell$ points tied by curves, the result is set to zero. For example
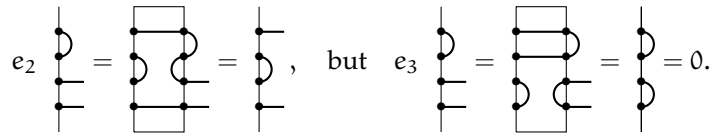
\begin{equation*}
	e_2\;\begin{tikzpicture}[baseline={(current bounding box.center)},scale=1/3]
	\draw (0,0) -- (0,5);
	\foreach \y in {1,2,3,4}
    		\node at (0,\y) {$\scriptstyle\bullet$};
	\begin{scope}[thick]
		\draw (0,1) -- (1,1);
		\draw (0,2) -- (1,2);
		\draw (0,3) arc [start angle = -90, end angle=90, x radius=0.5, y radius=0.5];
	\end{scope}
\end{tikzpicture}=\begin{tikzpicture}[baseline={(current bounding box.center)},scale=1/3]
	\draw (0,0) -- (2,0) -- (2,5) -- (0,5) -- (0,0);
	\foreach \y in {1,2,3,4}
    		\node at (0,\y) {$\scriptstyle\bullet$};
	\foreach \y in {1,2,3,4}
    		\node at (2,\y) {$\scriptstyle\bullet$};
	\begin{scope}[thick]
		\draw (0,1) -- (2,1);
		\draw (0,4) -- (2,4);
		\draw (2,2) -- (3,2);
		\draw (2,1) -- (3,1);
		\draw (0,2) arc [start angle=-90, end angle=90, x radius=0.5, y radius=0.5];
		\draw (2,2) arc [start angle=270, end angle=90, x radius=0.5, y radius=0.5];
		\draw (2,3) arc [start angle=-90, end angle=90, x radius=0.5, y radius=0.5];
	\end{scope}[thick]
\end{tikzpicture}=\begin{tikzpicture}[baseline={(current bounding box.center)},scale=1/3]
	\draw (0,0) -- (0,5);
	\foreach \y in {1,2,3,4}
    		\node at (0,\y) {$\scriptstyle\bullet$};
	\begin{scope}[thick]
		\draw (0,1) -- (1,1);
		\draw (0,4) -- (1,4);
		\draw (0,2) arc [start angle=-90, end angle=90, x radius=0.5, y radius=0.5];
	\end{scope}
\end{tikzpicture}\,,\quad\text{but}\quad e_3\;\begin{tikzpicture}[baseline={(current bounding box.center)},scale=1/3]
	\draw (0,0) -- (0,5);
	\foreach \y in {1,2,3,4}
    		\node at (0,\y) {$\scriptstyle\bullet$};
	\begin{scope}[thick]
		\draw (0,1) -- (1,1);
		\draw (0,2) -- (1,2);
		\draw (0,3) arc [start angle = -90, end angle=90, x radius=0.5, y radius=0.5];
	\end{scope}
\end{tikzpicture}=\begin{tikzpicture}[baseline={(current bounding box.center)},scale=1/3]
	\draw (0,0) -- (2,0) -- (2,5) -- (0,5) -- (0,0);
	\foreach \y in {1,2,3,4}
    		\node at (0,\y) {$\scriptstyle\bullet$};
	\foreach \y in {1,2,3,4}
    		\node at (2,\y) {$\scriptstyle\bullet$};
	\begin{scope}[thick]
		\draw (0,3) -- (2,3);
		\draw (0,4) -- (2,4);
		\draw (2,1) -- (3,1);
		\draw (2,2) -- (3,2);
		\draw (0,1) arc [start angle=-90, end angle=90, x radius=0.5, y radius=0.5];
		\draw (2,1) arc [start angle=270, end angle=90, x radius=0.5, y radius=0.5];
		\draw (2,3) arc [start angle=-90, end angle=90, x radius=0.5, y radius=0.5];
	\end{scope}[thick]
\end{tikzpicture}=\begin{tikzpicture}[baseline={(current bounding box.center)},scale=1/3]
	\draw (0,0) -- (0,5) -- (0,0);
	\foreach \y in {1,2,3,4}
    		\node at (0,\y) {$\scriptstyle\bullet$};
	\begin{scope}[thick]
		\draw (0,1) arc [start angle=-90, end angle=90, x radius=0.5, y radius=0.5];
		\draw (0,3) arc [start angle=-90, end angle=90, x radius=0.5, y radius=0.5];
	\end{scope}
\end{tikzpicture}=0.
\end{equation*}
That this action defines $\tl$-modules is shown in \cite{Westbury, RSA2012}. This result is independent of whether $q$ is generic or not.

For $q$ generic the irreducible modules $\mathcal V_{m}$ are in one-to-one correspondence with the $\mathsf V_{n,\ell}$:
$$\mathcal V_{m}\cong\mathsf V_{n,\ell=\frac n2-m},\qquad \textrm{for $q$ generic}.$$
Note that, when $q$ is a root of unity, the standard modules $\mathsf V_{n,\ell}$ are not irreducible in general, though they are always indecomposable. 

Several central elements (Casimir) can be used to distinguish the irreducible modules. An element $F_n\in\tl$ was shown to be central in \cite{AMDYSA2011}. Even though its explicit form will not be needed here, it is important to stress that it is a linear combination of words in $\tl$ with coefficients in $\mathbb Z[q,q^{-1}]$ and that, on the standard modules 
$\mathsf V_{n,\frac n2-m}$, it acts as
$$F_n\bigr|_{\mathsf V_{n,\frac n2-m}}=\bigl(q^{2m+1}+q^{-2m-1}\bigr)\mathds{1}.$$
As for the Casimir $S^2$ of $\uq$, the eigenvalues of the central element $F_n$ completely distinguish the irreducible modules when $q$ is generic, that is, if its eigenvalues on $\mathsf V_{n,\ell}$ and $\mathsf V_{n,\ell'}$ are equal, then $\mathsf V_{n,\ell}\cong\mathsf V_{n,\ell'}$.

\subsubsection{Representation theory of $\tl$ for $q$ a root of unity}\label{sec:tlnRoot}

As for that of $\uq$, the representation theory of the Temperley-Lieb algebra $\tl(q)$ for $q$ a root of unity is much richer than for $q$ generic.

At $q$ a root of unity, the algebra $\tl(q)$ is in general non-semisimple. More precisely if $q$ is a root associated with the integer $p$, then $\tl(q)$ is non-semisimple for all $n\ge p$, with one exception: For $p=2$ and $n$ odd, the algebra $\tl(\pm i)$ is semisimple.

The non-semisimplicity has an immediate consequence: There are representations of $\tl(q)$ that are indecomposable, but not irreducible. For example, the standard module $\mathcal V_m\cong \mathsf V_{n,\frac n2-m}$ is not irreducible in general, though it remains always indecomposable. We shall denote by $\mathcal I_j\cong \mathcal V_j/\mathcal R_j$ its irreducible quotient, where $\mathcal R_j$ is its (unique) maximal proper submodule. If $\mathcal R_j$ is trivial, then $\mathcal V_j$ is irreducible. (See \cite{Martin1991,RSA2012}.) More importantly, the algebra itself, seen as a left $\tl$-module, is not a sum of irreducible ones. The indecomposable modules appearing in its decomposition are called the principal indecomposable modules. They are projective covers of the irreducible modules and we shall denote by $\mathcal P_j$ the principal indecomposable whose irreducible quotient is the module $\mathcal I_j$. If $j$ is critical, then the corresponding projective is irreducible and the following three modules coincide: $\mathcal P_j\cong \mathcal V_j\cong \mathcal I_j$. If $j'<j$ are two consecutive elements of a (non-critical) orbit (see paragraph \ref{subsect:ru}), then the projective $\mathcal P_j$ is part of a non-split exact sequence
$$0\longrightarrow \mathcal V_{j'}\longrightarrow \mathcal P_j\longrightarrow \mathcal V_j\longrightarrow 0,$$
that is, $\mathcal P_j$ has $\mathcal V_{j'}$ as one of its proper submodules and $\mathcal V_j\cong \mathcal P_{j}/\mathcal V_{j'}$, even though $\mathcal P_j$ is not a direct sum of the two modules $\mathcal V_{j'}$ and $\mathcal V_{j}$. (For more details see \cite{Martin1991, GoodmanWenzl1993} and, for a presentation closer to the graphical description used in paragraph \ref{sec:ton}, \cite{Westbury, RSA2012}.)

There are more (finite) indecomposable $\tl$-modules beside the principal $\mathcal P_j$'s and the standard $\mathcal V_j$ ones. Fortunately, as will be recalled in the next paragraph \ref{sec:rho}, only principal and standard modules appear in the decomposition of $\cn$.

\subsubsection{The representation of $\tl$ on $\cn$}\label{sec:rho}

A representation on the tensor product space $\cn$ for $n\ge 2$ is given by the algebra homomorphism $\rho_n:\tl\to\mathrm{End}\cn$ defined on generators by
\begin{equation}
	\label{eq:TLnrep}
	\rho_n(e_i) = \mathds{1}_2\otimes\dotsm\otimes\mathds{1}_2\otimes E\otimes\mathds{1}_2\otimes\dotsm\otimes\mathds{1}_2,
\end{equation}
where the matrix
\begin{equation}
	\label{eq:matE}
	E = \begin{pmatrix} 0 & 0 & 0 & 0 \\ 0 & q^{-1} & -1 & 0  \\ 0 & -1 & q & 0 \\ 0 & 0 & 0 & 0 \end{pmatrix}
\end{equation}
takes up positions $i$ and $i+1$ in the above tensor product and there are therefore $(n-2)$ factors $\mathds{1}_2$. As for $\uq$, wherever the context is clear enough, we shall omit the writing of $\rho_n$, e.g.~$e_i\,v$ will mean $\rho_n(e_i)\,v$.

The action \eqref{eq:TLnrep} of $e_i$ on vectors of the spin basis $\mathcal B$ is ``local'' as it changes only the $i$-th and $(i+1)$-th spins. Moreover the matrix $E$ does not change the number of ``$+$'' or ``$-$'' in such a vector. Therefore the eigensubspace of $S^z=\frac12\sum_{1\le i\le n}\sigma^z_i$ are $\tl$-submodules and the decomposition $\cn = \oplus_{-n/2\leq m \leq n/2}W_m$ introduced in \eqref{eq:decompocn} holds as a direct sum of $\tl$-modules. Notice that the modules $W_m$ and $W_{-m}$ are isomorphic. The isomorphism is given by the spin-reversal operator $R=\otimes^n\sigma_x$, where $\sigma_x=\left(\begin{smallmatrix} 0 & 1 \\ 1 & 0 \end{smallmatrix}\right)$, coupled with the inversion $q\mapsto q^{-1}$. This operation manifestly commutes with the action of $\tl$ given by \eqref{eq:matE}. 
It is therefore sufficient to restrict the analysis to modules $W_m$ with $m\geq0$.

The eigenspace $W_m$ of $S^z$ contains a subspace isomorphic to $\mathcal V_{m}$ as $\tl$-module. A map $\psi_{n,m}:\mathcal V_{m}\cong\mathsf V_{n,p=\frac n2-m}\rightarrow W_m$ is constructed as follows. For a link state in $\mathsf V_{n,p}$, let $\{(i_1,j_1),(i_2,j_2),\dots,\allowbreak (i_p,j_p)\}$ be the pairs of points in $\{1,2,\dots, n\}$ with $i_k<j_k$ that are pairwise connected. The link state is mapped to $\prod_{1\le k\le p}T(i_k,j_k)\ket{++\dots+}$ where $T(i_k,j_k)=q^{-\frac12}\sigma_{i_k}^--q^{\frac12}\sigma_{j_k}^-$. For example the link state $\begin{tikzpicture}[scale=0.1]]
\useasboundingbox (-0.5,1.5) rectangle (2,4);
    \draw[line width=0.5pt] (0,0) -- (0,5);
    	\draw[line width=0.5pt] (0,1) arc (270:450:1.5);
    	\draw[line width=0.5pt] (0,2) arc (270:450:0.5);
\end{tikzpicture}$ in $\mathsf{V}_{4,2}$ is mapped to
$$q\ket{++--}-\ket{+-+-}-\ket{-+-+}+\frac1q\ket{--++}.$$
The verification that $\psi_{n,m}$ is a $\tl$-homomorphism is straightforward.

The decomposition of $\cn$ as $\tl$-module is rather simple when $q$ is generic. (In this case, the decomposition will be given a new proof in Corollary \ref{cor:isoMathcalV}.) To our knowledge the decomposition for $q$ a root of unity has been worked out first by \citet{Martin1992}. To present it, we shall use the recent description of \cite{GV2012}. Suppose $q$ is a root associated with $p$ and write $n=r_mp+s_m$ with $r_m\in\mathbb N$ and $-1\le s_m\le p-2$. Then
\begin{align}\label{eq:GV}
\cn&\cong 
\Big(\bigoplus_{1\le r\le r_m-1}\bigoplus_{\substack{0\le s\le p-1\\[1mm] rp+s+n\equiv 1\textrm{\,mod\,}2}}r(p-s)\cdot\mathcal P_{(rp+s-1)/2}\Big)
\oplus
\Big(\bigoplus_{\substack{0\le s\le s_m+1\\[1mm] s+s_m\equiv 1\textrm{\,mod\,}2}} r_m(p-s)\cdot\mathcal P_{(r_mp+s-1)/2} \Big)\notag \\
&\ \quad
\oplus
\Big(\bigoplus_{\substack{1\le s\le s_m+1\\[1mm] s+s_m\equiv 1\textrm{\,mod\,}2}}(r_m+1)s \cdot\mathcal V_{(r_mp+s-1)/2}\Big)
\oplus
\Big(\bigoplus_{\substack{s_m+2\le s\le p-1\\[1mm] s+s_m\equiv 1\textrm{\,mod\,}2}} r_m(p-s)\cdot\mathcal V_{(r_mp-s-1)/2}\Big)
\end{align}
as $\tl$-module. The integer $r_m$ is the number of critical lines falling on the rightmost $j=\frac n2$ or to its left. The first sum contains all principal indecomposable modules $\mathcal P_j$ for $j$'s that fall to the left of the rightmost critical line and the second those $\mathcal P_j$'s that lie to its right. The last two sums contain standard modules $\mathcal V_j$ with $j$ in the window to the left of the last critical line (third sum) and to its right (fourth sum). The integer that multiplies the principal and standard modules in the sums, like the factor $r(p-s)$ that appears in the first, is the number of isomorphic copies of these modules, that is their {\em multiplicities} in the decomposition of $\cn$. It is possible to rewrite this decomposition in terms of the orbits introduced earlier. It then reads
\begin{align}\label{eq:GY}
\cn&\cong \Big(\bigoplus_{\substack{\textrm{non-critical}\\[1mm]\text{orbits orb}_j}}\ \bigoplus_{\substack{j_i\in\,\text{orb}_j\\[1mm] i\ge 2}}(i-1)(ip-2j_i-1)\cdot \mathcal P_{j_i}\Big) 
\oplus\Big(\bigoplus_{j\textrm{\ critical}}(2j+1)\cdot\mathcal P_j\Big)\notag\\
&\ \qquad
\oplus\Big(\bigoplus_{\substack{j_l\in\,\text{orb}_j\\[1mm] j\textrm{\ noncritical}}}l(2j_l+1-(l-1)p)\cdot \mathcal V_{j_l}\Big)
\end{align}
where, in the last sum, the index $j_l$ stands for the last element of the orbit $orb_j$. All the modules appearing in either \eqref{eq:GV} or \eqref{eq:GY} are indecomposable and the main objective of this paper is to find the primitive idempotents projecting on each. These idempotents are found by exploiting the duality between the Temperley-Lieb algebra and the quantum algebra.

\subsection{The quantum Schur-Weyl duality}\label{subsect:qSW}
\begin{defn}[Hecke algebra]
	\label{defn:Hn}
	For $q\in\C^\times$, the unital associative algebra over $\C$ generated by $\{h_1,\dotsc,h_{n-1}\}$ and satisfying
	\begin{align}\label{eq:Hnrelations}
		\begin{split}
			h_i^2			& = (q-q^{-1})\,h_i+1,			\\
			h_i\,h_{i+1}\,h_i 	& = h_{i+1}\,h_i\,h_{i+1},			\\
			h_i\,h_j			& = h_j\,h_i,\hspace{15mm}|i-j| > 1,
		\end{split}
	\end{align}
	is called the \emph{Hecke algebra} $\hn(q)$.
\end{defn}
This algebra is known as the $q$-deformation of the group algebra $\C\mathrm{S}_n$, where $\mathrm{S}_n$ is the symmetric group of $n$ elements. In the limit $q\to1$, the relations \eqref{eq:Hnrelations} become those of $\C\mathrm{S}_n$. For $q$ generic the two algebras $\hn(q)$ and $\C\mathrm{S}_n$ are isomorphic \cite{Wenzl1988}.

The Temperley-Lieb and Hecke algebras are related by a surjective homomorphism
\begin{equation}
	\label{eq:HnToTLn}
	\begin{gathered}
		\phi:\hn(q)\to\tl(q)								\\
		h_i\mapsto e_i-q^{-1}1_{\tl}\quad\text{and}\quad1_{\hn}\mapsto1_{\tl}.
	\end{gathered}
\end{equation}
There is a representation $\sigma_n:\hn(q)\to\mathrm{End}\cn$ such that
\begin{equation}
	\label{eq:Hnrep}
	\sigma_n(h_i)=(\rho_n\circ\phi)(h_i) = \mathds{1}_2\otimes\dotsm\otimes\mathds{1}_2\otimes H\otimes\mathds{1}_2\otimes\dotsm\otimes\mathds{1}_2,
\end{equation}
where the matrix $H$ takes up positions $i$ and $(i+1)$ and is 
\begin{equation*}
	H = E-q^{-1}= \begin{pmatrix} -q^{-1} & 0 & 0 & 0 \\ 0 & 0 & -1 & 0 \\ 0 & -1 & q-q^{-1} & 0 \\ 0 & 0 & 0 & -q^{-1} \end{pmatrix}.
\end{equation*}

If $A$ is an algebra and $S\subset A$, the \emph{centralizer} of $S$ is defined as $C_A(S)=\{a\in A:sa=as,\ \forall s\in S\}$. Now, if $M$ is an $A$-module and $\mu:A\to\mathrm{End}\,M$ is the corresponding algebra homomorphism, then we have $C_{\mathrm{End}\,M}\bigl(\mu(A)\bigr)=\mathrm{End}_{\mu(A)}M$. The notation $\mathrm{End}_{\mu(A)}M$ stands for the algebra of endomorphisms of $M$ that commute with the action of $A$ on $M$. (The mathematical literature will drop any reference to $\mu$, writing simply $\mathrm{End}_{A}M$.)

\begin{thm}[Quantum Schur-Weyl duality]
	\label{thm:qSchurWeyl}
	Let $V=\cn$, $A=\mathrm{End}\,V$, $S_U=\pi_n\bigl(\uq\bigr)$ and $S_H=\sigma_n\bigl(\hn(q)\bigr)$, where $\pi_n$, $\sigma_n$ are as above. For $q$ generic, the two subalgebras $S_U$ and $S_H$ of $A$ are mutual centralizers: 
	\begin{equation}
		\label{eq:qSchurWeylResult}
		S_U=\mathrm{End}_{S_H}V\qquad\text{and}\qquad S_H=\mathrm{End}_{S_U}V.
	\end{equation}
\end{thm}
\noindent This result is sometimes called the \emph{$q$-Schur-Weyl duality}, by analogy with the Schur-Weyl duality between $\C\,\mathrm{S}_n$ and $\mathrm{U}\,\mathfrak{sl}_2$, corresponding to $q\to1$. It is due to \citet{Jimbo1986}. One immediate consequence is that $\pi_n(S^+)$ and $\pi_n(S^-)$ act on $\cn$ as $\tl$-homo\-morphisms. 
 
The set of matrices of $\mathrm{End}\,V$ that commute with the generators $g_H=\bigl\{\sigma_n(h_1),\dotsc,\sigma_n(h_{n-1})\bigr\}$ are commuting with the full operator algebra $\sigma_n(\hn)$. The same is true for $g_{TL}=\bigl\{\rho_n(e_1),\allowbreak \dotsc,\allowbreak\rho_n(e_{n-1})\bigr\}$. Because of the homomorphism $\phi$ and the fact that $\sigma_n=\rho_n\circ \phi$, the matrices in the two sets $g_H$ and $g_{TL}$ are equal up to an additive multiple of the identity. The set of matrices commuting with $g_H$ is therefore equal to the one for $g_{TL}$. It follows that $\mathrm{End}_{\sigma_n(\hn)}V=\mathrm{End}_{\rho_n(\tl)}V$ and the Schur-Weyl duality implies
\begin{equation}
	\label{eq:qSchurWeyl}
	\mathrm{End}_{\tl}\!\otimes^n\C^2 = \pi_{n}\bigl(\uq\bigr),
\end{equation}
where $\mathrm{End}_{\tl}$ is a shorthand notation for $\mathrm{End}_{\rho_n(\tl)}$.

The Schur-Weyl duality proved by Jimbo was extended by \citet{Martin1992} for $q$ a root of unity. We shall use his general result in the following case. Let $q$ be a root of unity and let ${\mathcal L}U$ be the algebra $\uq$ extended by the divided powers $(S^\pm)^{(p)}$. Then, if again $S_{{\mathcal L}U}$ and $S_{TL}$ stand for $\pi_n({\mathcal L}U)$ and $\rho_n(\tl(q))$, $S_{{\mathcal L}U}=\mathrm{End}_{S_{TL}}\cn$ and $S_{TL}=\mathrm{End}_{S_{{\mathcal L}U}}\cn$.

Recall from section \ref{sect:algebraTL} that $\cn$, seen as a $\tl(q)$-module, decomposes naturally as 
\begin{equation*}
	\cn=W_{-n/2}\oplus W_{-n/2+1}\oplus\dotsb\oplus W_{n/2}.
\end{equation*}
The endomorphisms on $W_m\subset \otimes^n\C^2$ that can be created out of $\pi_n\bigl(\uq\bigr)$ must be linear combinations of $(S^-)^r(S^+)^r$, $(S^+)^s(S^-)^s$
and of $q^{\pm S^z}$. The latter acts as a multiple of the identity on $W_m$ and, by an argument of the type leading to the Poincar\'e-Birkhoff-Witt theorem (\cite{KS1997, ChariPressley}), it is sufficient to restrict the linear combinations to the $(S^-)^r(S^+)^r$ with $0\leq r\leq n/2-m$. The upper bound $n/2-m$ is the number of ``$-$'' signs in each element of the basis $\mathcal B$ that also belongs to $W_m$. If $r$ exceeds this number, then $(S^+)^r$ annihilates $W_m$. Because $(S^\pm)^p=0$ if $q$ is a root of unity associated with the integer $p$, it will be useful for section \ref{sec:atRootsOfUnity} to consider instead combinations of the nonzero (and well defined) 
$$S_r=(S^-)^{(r)}(S^+)^{(r)},\qquad \textrm{with } S_0=\mathds{1}.$$ 
The multiplication of two of these, restricted to $W_m$, is given by
\begin{equation*}
	\left.S_kS_l\right|_{W_m} = \sum_{i=0}^kC^m_{k,l,i}\,\left.S_{l+i}\right|_{W_m},\quad\text{for }l\geq k,
\end{equation*}
with structure constants given by $C^m_{k,l,i} = \qbinmini{l+i}{k}\qbinmini{l+i}{l}\qbinmini{2m+k+l}{k-i}$. (See Proposition \ref{lem:sksl} in the Appendix.) The set of endomorphisms $\{S_0,\allowbreak S_1,\dotsc,S_{n/2-m}\}$ is therefore closed under multiplication and thus generates an algebra which is found to be abelian (see Lemma \ref{lem:commutS}). Finally, for all $q$
\begin{equation*}
	\mathrm{End}_{\tl}W_m=\mathrm{span}\bigl\{S_0,S_1,\dotsc,S_{n/2-m}\bigr\}.
\end{equation*}

In the two following sections, we tackle the problem of finding the primitive idempotents for $\cn$, viewed as a $\tl$-module, for any $q$.
%
\section{Decomposition of $\cn$ for $q$ generic}\label{sect:decqgene}
%
We start by a few observations that also indicate how the idempotents were discovered.

When $q$ is generic, the action of $S_r$, $0\leq r\leq n/2-m$, on a vector $\ket{j,m}_k$, where $m\leq j\leq n/2$ and $1\leq k\leq\Gamma_j$, is diagonal (Proposition \ref{prop:Sraction}):
\begin{equation}
	\label{eq:Sraction}
	S_r\ket{j,m}_k = \qbin{j+m+r}{r}\qbin{j-m}{r}\ket{j,m}_k.
\end{equation}
Recall that an idempotent $z$ in $\mathrm{End}_{\tl}W_m$ is a nonzero endomorphism on $W_m$ such that $z^2=z$. Section \ref{subsect:qSW} has shown that these endomorphisms are expressible as linear combinations of the $S_r$'s. The action of the idempotents on $W_m$ is therefore diagonal in the basis $\{\ket{j,m}_k, j\ge m, k=1,\dots, \Gamma_j\}$ and any nonzero linear combination of the $S_r$'s with only $0$'s and $1$'s on the diagonal is such an idempotent. By equation \eqref{eq:Sraction}, $S_r$ acts the same way on all isomorphic copies of $U_j$ in $\otimes^n\mathbb C^2$. The smallest subspace upon which a linear combination $z=\sum_r a_rS_r$ can project must contain the sum $V_{j,m}$ of all eigensubspaces where $S_z=m\cdot\mathds{1}$ of isomorphic copies of $U_j$ contained in $\otimes^n\mathbb C^2$. Since $U_j$ and $U_{j'}$ are non-isomorphic for distinct $j$ and $j'$, the subspace $W_m$ splits into
\begin{equation}\label{eq:wmAndVjm}W_m=\bigoplus_{j\ge m}V_{j,m}
\end{equation}
as vector space. Therefore we look for primitive idempotents of the form
\begin{equation}
	\label{eq:idem}
	z^{(n)}_{j,m}=\sum_{i=0}^{n/2-m}a_{i,j,m}\,S_i.
\end{equation}

Suppose now that the idempotents for $n-2$ have been constructed. The problem of finding idempotents for $n$ requires no new values for $j$ besides $j=n/2$. The action \eqref{eq:Sraction} does not depend on $n$ and thus implies that
\begin{equation}
	\label{eq:recursive}
	z^{(n)}_{j,m}=z^{(n-2)}_{j,m} + a_{n/2-m,j,m}\,S_{n/2-m}.
\end{equation}
The new term $S_{n/2-m}$ does not change the action of the projectors on subspaces $V_{j',m}$ for any $j'$ that appears in ${\otimes^{n-2}\C^2}$, but is necessary to describe properly their action on the subspace $V_{n/2,m}$ appearing in $\cn$.

Assume now that $z^{(n)}_{j,m}$ projects on a subspace (containing) $V_{j,m}$ for $j<n/2$. It must act as $0$ on $\ket{n/2, m}$:
\begin{equation*}
	\sum_{i=0}^{n/2-m}a_{i,j,m}\,S_i\ket{n/2,m} = \sum_{i=0}^{n/2-m}a_{i,j,m}\qbin{n/2+m+i}{i}\qbin{n/2-m}{i}\ket{n/2,m} = 0.
\end{equation*}
We may then express the coefficient $a_{n/2-m,j,m}$ in terms of the other coefficients:
\begin{equation*}
	a_{n/2-m,j,m} = -\qbin{n}{n/2-m}^{-1}\times\sum_{i=0}^{n/2-m-1}\qbin{n/2+m+i}{i}\qbin{n/2-m}{i}a_{i,j,m}.
\end{equation*}
If the coefficients for $n'$ smaller than $n$ are known, then equation \eqref{eq:recursive}, together with the preceding one, gives the expression of the idempotent $z_{j,m}^{(n)}$. Some exploration (and guessing) allowed us to solve the recursion:
\begin{equation}
	\label{eq:coefExpr}
	a_{i,j,m}=(-1)^{i+j-m} \qbin{i}{j-m}\qbin{i+j+m+1}{i+1}^{-1}\frac{[2j+1]}{[i+1]}.
\end{equation}
The left $q$-binomial vanishes for $i<j-m$. Therefore, the summation in \eqref{eq:idem} may be truncated to $j-m\le i\le n/2-m$. The next theorem proves that the element $z_{j,m}^{(n)}$ defined using these coefficients is indeed idempotent, primitive and projects on $V_{j,m}$. The notation will be lightened up by the omission of the superscript ``$(n)$'' whenever possible.

\begin{thm}[Primitive idempotents for $q$ generic]
	\label{thm:idqgene}
	Let $q$ be generic. The elements $\{z_{j,m}\}_{m\leq j\leq n/2}$ defined as
	\begin{equation}
		\label{eq:idem2}
		z_{j,m} = \sum_{i=j-m}^{n/2-m}a_{i,j,m}\,S_i
	\end{equation}
with coefficients \eqref{eq:coefExpr}, constitute a set of mutually orthogonal primitive idempotents that partitions unity in $\mathrm{End}_{\tl}W_m$, that is $\sum_{m\le j\le n/2}z_{j,m}=\mathds{1}_{W_m}$. Furthermore, the $\tl$-modules $V_{j,m}=z_{j,m}W_m$ are irreducible.
\end{thm}
\begin{proof}
	Acting on a vector $\ket{j,m}$ of one of the irreducible $U_j$'s, the element $z_{j,m}$ yields
	\begin{equation*}
		z_{j,m}\ket{j,m} = \sum_{i=j-m}^{n/2-m}a_{i,j,m}\,S_i\ket{j,m} = \sum_{i=j-m}^{n/2-m}a_{i,j,m}\qbin{j+m+i}{i}\qbin{j-m}{i}\ket{j,m}.
	\end{equation*}
	If $i>j-m$ then $\qbinmini{j-m}{i}=0$, and the last equation becomes
	\begin{equation*}
		z_{j,m}\ket{j,m} = a_{j-m,j,m}\,S_{j-m}\ket{j,m}=\ket{j,m}
	\end{equation*}
	after simplification. Therefore $z_{j,m}$ acts as the identity on $V_{j,m}$. Now, the action of $z_{j,m}$ on $\ket{j',m}$ with $j'\neq j$ is
	\begin{equation*}
		z_{j,m}\ket{j',m} = \sum_{i=j-m}^{n/2-m}a_{i,j,m}\qbin{j'+m+i}{i}\qbin{j'-m}{i}\ket{j',m}.
	\end{equation*}
	If $j'<j$, then $i>j'-m$ for all $i$, and consequently $z_{j,m}\ket{j',m}=0$. If $j'>j$, then $\bigl[\begin{smallmatrix} j'-m \\ i \end{smallmatrix}\bigr]=0$ for $i>j'-m$, but it is nonzero for $j-m\leq i\leq j'-m$. To show that $z_{j,m}\ket{j',m}=0$ for all $j'\neq j$, it remains to see 	whether
	\begin{equation*}
		\sum_{i=j-m}^{j'-m}a_{i,j,m}\qbin{j'+m+i}{i}\qbin{j'-m}{i} = 0
	\end{equation*}
holds. Let $j'=j+k$ with $k>0$ and expand the sum as
	\begin{multline*}
		 \sum_{i=j-m}^{j+k-m}a_{i,j,m}\qbin{j+k+m+i}{i}\qbin{j+k-m}{i} = \\[2mm] [2j+1]\,\frac{[j+m]!\,[j+k-m]!}{[j-m]!\,[j+k+m]!}\;\sum_{i=j-m}^{j+k-m}(-1)^{i+j-m}\,\frac{[j+k+m+i]!}{[i-j+m]!\,[i+j+m+1]!\,[j+k-m-i]!}.
	\end{multline*}
The inner sum takes the form of the series $A_k$ of Proposition \ref{prop:sumqnb} (a) if $r$ is set to  $i-j+m$:
	\begin{align*}
		A_k	& = A_{k-1} + (-1)^k\frac{[2j+2k]!}{[k]!\,[2j+k+1]!}										\\[2mm]
			& = (-1)^{k-1}\frac{[2j+2k]!}{[k]!\,[2j+k+1]!} + (-1)^k\frac{[2j+2k]!}{[k]!\,[2j+k+1]!}	\\[2mm]
			& = 0.
	\end{align*}
	Up to now, we have shown that $z_{j,m}\ket{j',m} = \delta_{j,j'}\ket{j',m}$ and therefore $z_{j,m}\,W_m = V_{j,m}$. Clearly the idempotence and orthogonality relations 
	\begin{equation}
		\label{eq:orthoIdem}
		z_{j,m}\,z_{j',m}=\delta_{j,j'}\,z_{j',m}
	\end{equation}
as well as $\sum_{m\le j\le n/2}z_{j,m}=\mathds{1}_{W_m}$ follow because of \eqref{eq:wmAndVjm}. Those idempotents are primitive as the set $\{z_{j,m}\}_{m\leq j\leq n/2}$ is a basis of $\mathrm{End}_{\tl}W_m$. Indeed the idempotents are linearly independent as they project onto mutually disjoint subspaces and their number $(n/2-m+1)$ coincides with the dimension of $\mathrm{End}_{\tl}W_m$ (see the end of paragraph \ref{subsect:qSW}).

Finally the submodules $V_{j,m}$ are irreducible. Indeed, if $V_{j,m}$ has a proper submodule, then the endomorphism ring $\mathrm{End}_{\tl}V_{j,m}$ must contain at least one element linearly independent from the unit $z_{j,m}$. But, by a simple dimensional argument ($\dim \mathrm{End}_{\tl} W_m=n/2-m+1$), the endomorphism algebra of $V_{j,m}$, namely $z_{j,m}\mathrm{End}_{\tl}W_m$, is spanned by its unit $z_{j,m}$.
\end{proof}

We now present three corollaries of the previous theorem. The first one establishes a link between $\mathcal{V}_{j}$ and $V_{j,m}$. This result is well-known, but the idempotents provide a new simple proof.
\begin{cor}
	\label{cor:isoMathcalV}
	The irreducible module $V_{j,m}=z_{j,m}W_m$ is isomorphic to $\mathcal{V}_{j}$ and $W_m\cong\oplus_{m\le j\le \frac n2}\mathcal V_{j}$ as $\tl$-modules.
\end{cor}
\begin{proof}
The subspace $W_{\frac n2}=\mathrm{span}\,\{\ket{++\dots +}\}$ is one-dimensional and all generators $\rho_n(e_i)$ act on it as zero. This is precisely their action on the one-dimensional irreducible $\mathcal V_{\frac n2}\cong\mathsf V_{n,0}$. As required, equation \eqref{eq:idem2} gives $z^{(n)}_{\frac n2,\frac n2}=\mathds{1}_{W_{\frac n2}}$. Therefore $z^{(n)}_{\frac n2,\frac n2}$ projects on a one-dimensional subspace and the decomposition of $W_{\frac n2}\cong\mathcal V_{\frac n2}$ as $\tl$-module follows.

	Suppose now that the decomposition $W_{m+1}\cong\oplus_{m+1\le j\le \frac n2}\mathcal V_{j}$ holds for some $m\geq 0$ and that all $V_{j,m+1}=z^{(n)}_{j,m+1}W_{m+1}$ are isomorphic to the corresponding irreducible $\mathcal V_{j}$. For $q$ generic, the action of $S^-$ restricted to $W_{m+1}$ is injective in any $U_j$. Therefore, the irreducible representation $\mathcal V_{j}\subset W_{m+1}$ is mapped by $S^-$ into a subspace transforming also as $\mathcal V_{j}$ and $W_m$ must therefore contain a $\tl$-submodule $\oplus_{m+1\le j\le \frac n2}\mathcal V_{j}$. Since $V_{j,m+1}=z^{(n)}_{j,m+1}W_{m+1}$ is isomorphic to $\mathcal V_{j}$ and that the non-zero $\tl$-homomorphism $S^-$ maps the vector space $V_{j,m+1}$ onto $V_{j,m}$, then $V_{j,m}=z^{(n)}_{j,m}W_m$ is also isomorphic to $\mathcal V_{j}$ as $\tl$-module. Finally section \ref{sec:rho} has shown that $W_m$ always contains a subspace isomorphic to $\mathcal V_{m}$. This module is non-isomorphic to those contained in $S^-W_{m+1}$ and $W_m$ must therefore contain a submodule $\oplus_{m\le j\le \frac n2}\mathcal V_{j}$. The dimension of this submodule is $\sum_{m\le j\le \frac n2}\Gamma^{(n)}_j=\left(\begin{smallmatrix}n\\ \frac n2-m\end{smallmatrix}\right)$ and coincides with $\dim W_m$. Since all idempotents $z^{(n)}_{j,m}$ for $m+1\le j\le \frac n2$ have been accounted for, the last one $z^{(n)}_{m,m}$, which is non-zero, must project onto $\mathcal V_{m}$.
\end{proof}

The second corollary is an identity between $q$-binomials.
\begin{cor} For generic $q$ and $2m\in \mathbb N$ with $j-m\in \mathbb N$
	\begin{equation*}
		\sum_{j=m}^{m+i}(-1)^{j-m}[2j+1]\qbin{i}{j-m}\qbin{i+j+m+1}{i+1}^{-1}=0,\qquad \textrm{for all\ }i\ge 1.
	\end{equation*}
\end{cor}
\begin{proof}
The identity $\sum_{m\le j\le n/2}z^{(n)}_{j,m}=\mathds{1}_{W_m}$ can be written as
	\begin{equation*}
		\sum_{i=0}^{\frac n2-m}S_i\sum_{j=m}^{m+i}a_{i,j,m}=\mathds{1}_{W_m},
	\end{equation*}
since $a_{i,j,m}=0$ for $j>i+m$. The $S_i$ are linearly independent as linear transformations and $S_0\,a_{0,j,m}=\mathds{1}_{W_m}$. Therefore $\sum_{m\le j\le m+i}a_{i,j,m}$ must vanish for $i\geq1$. 
\end{proof}

\noindent This combinatorial identity can also be proved directly. By fixing $r=j-m$ and using \eqref{eq:coefExpr}, one can write the sum (up to an overall factor) as
	\begin{equation*}
		 \sum_{r=0}^i(-1)^r\frac{[2m+r]!\,[2m+2r+1]}{[r]!\,[i-r]!\,[2m+r+i+1]!} = B_i.
	\end{equation*}
Proposition \ref{prop:sumqnb} (b) then leads to the result for any $i\geq1$ :
	\begin{align*}
		B_i 	& = B_{i-1} + (-1)^i\frac{[2m+i]!}{[i]!\,[2m+2i]!}							\\[2mm]
			& = (-1)^{i-1}\frac{[2m+i]!}{[i]!\,[2m+2i]!} + (-1)^i\frac{[2m+i]!}{[i]!\,[2m+2i]!} 	\\[2mm]
			& = 0.
	\end{align*}
The similarity of this argument with that using the series $A_k$ in the proof of Theorem \ref{thm:idqgene} is striking.

The last corollary relates two central elements, the first in $\tl$, the second in $\uq$. It is particularly useful as it holds for any value of $q$. We exceptionally reinstate the ``$\rho_n$'' and ``$\pi_n$'' to underline that the result holds on $\cn$.
\begin{cor}\label{coro:fn} For all $q$
	$$\rho_n(F_n)=\pi_n\bigl((q-q^{-1})S^2+[2]\mathds{1}\bigr).$$
\end{cor}
\begin{proof}
	The first step is to show that the relation holds when $q$ is generic. On a given $\mathcal V_{n/2-j}$, the element $F_n$ acts as the identity times $\left(q^{2j+1}+q^{-2j-1}\right) = {[2(2j+1)]}/{[2j+1]}$. (See section \ref{sec:ton}.) The proof of Corollary \eqref{cor:isoMathcalV} has shown that all copies of a given $\mathcal V_{m}$ are obtained from the copy that lies in $W_m$ by the action of $S^-$. Since $S^2$ commutes with $S^-$, the action of $\left((q-q^{-1})S^2+[2]\mathds{1}\right)$ might as well be computed on any of the vectors $\ket{j,j}_{i=1,\dotsc,\Gamma_j}$:
	\begin{equation}
		\label{eq:Xaction}
		\left((q-q^{-1})S^2+[2]\mathds{1}\right)\ket{j,j} = \left((q-q^{-1})^2\bigl([j+1/2]^2-[1/2]^2\bigr)+[2]\right)\ket{j,j}	 = \frac{[2(2j+1)]}{[2j+1]}\ket{j,j}.
	\end{equation}
	The relation thus holds for generic $q$. The central element $F_n$ is a linear combination of words in $\tl$ with coefficients in $\mathbb Z[q,q^{-1}]$. Moreover, in the spin basis, the generators $e_i$ are represented by matrices whose elements are also polynomials in $q$ and $q^{-1}$. So $\rho_n(F_n)$ is a polynomial in $q$ and $q^{-1}$. The matrix elements of the Casimir $S^2$ of $\uq$ in the spin basis are also polynomials in $q$ and $q^{-1}$. If the two polynomials coincide on the open set of generic $q$'s, they coincide everywhere and the result must hold for all $q$.
\end{proof}
%
\section{Decomposition of $\cn$ for $q$ a root of unity}\label{sec:atRootsOfUnity}
%
The goal of this section is to study the behavior, when $q$ goes to a root of unity, of the idempotents $z_{j,m}$ found for generic $q$ in the previous section. It is known that, at $q$ a root of unity, the Temperley-Lieb algebra $\tl$ is non-semisimple, at least for $n$ large enough. Indecomposable representations exist at such values of $q$ and some of the previous idempotents may fail to exist. To identify the proper idempotents on $\cn$, we are guided by the {\em evaluation principle} stated by \citet{GoodmanWenzl1993}. Though obvious, it allows for the identification of proper quantities in $\tl$ at root of unity: Any algebraic identity between elements in $\tl$ that have as coefficients rational functions whose denominators do not have a zero at $q_c$ is an algebraic identity of $\tl$ at $q_c$. 

For the rest of the section, the integer $n$ appearing in $\tl$ and $\cn$ is fixed and the complex number $q_c$ is a root of unity associated with the integer $p$.
\subsection{The singularities at a root of unity $q_c$}
The present subsection gives the criteria for a coefficient $a_{i,j,m}$ of the idempotent $z_{j,m}$ to be singular. It is natural to write the defining indices $i,j,m$ as
\begin{align*}
	i & =  r\cdot p+a,  				\\
	j & =  s\cdot p+b,				\\
	\text{and}\quad m & =  t\cdot p+c,
\end{align*}
with $0\le r,s,t$ and $0\le a,b,c\le p-1$. However the following variables and their factorisation will be more useful:
\begin{equation}
	\label{eq:adg}
	\begin{aligned}
		 i 						& = r\cdot p+a, 		\\
		k = j-m 					& = u\cdot p+d,		\\
		\text{and}\quad i+l+1	= i+j+m+1	& = w\cdot p+g
	\end{aligned}
\end{equation}
where $0\le r, u, w$ and $0\le a, d, g\le p-1$. Recall that, when $n$ is odd, both $j$ and $m$ are half-integers. The labels $a$, $d$ and $g$ are however integers for all $n$. The expression for $a_{i,j,m}$ takes the following form in terms of $i, k$ and $l$:
\begin{equation*}
	a_{i,j,m}=(-1)^{i+k}\qbin{i}{k}\qbin{i+l}{i}^{-1}\frac{[k+l+1]}{[i+l+1]}.
\end{equation*}
As before the coefficient $a_{i,j,m}$ is zero if $i<k=j-m$. We now study the behavior of those coefficients as functions of $q$.

\begin{lem}\label{thm:singu1} Let $q_c$ and $p$ be as above. 
	\begin{enumerate}
		\item[(1)] If $j$ is critical, then the coefficient $a_{i,j,m}$ is regular at $q_c$ for all values of $i$ and $m$.
		\item[(2)] If $j$ is non-critical, then $a_{i,j,m}$ is singular at $q_c$ if and only if $g\le a$ and $d\le a$.
	\end{enumerate}
\end{lem}

\begin{proof} The proof makes systematic use of the $q$-Lucas theorem recalled in the Appendix as Lemma \ref{thm:LemFondamental}. A singularity in the above form for $a_{i,j,m}$ may appear only in the factors
	\begin{equation}\label{eq:facteur0}
		\qbin{i+l}{i}^{-1}\frac{1}{[i+l+1]}.
	\end{equation}
	With the use of \eqref{eq:adg}, the $q$-binomial can be written as
	\begin{equation}\label{eq:facteur1}
		\qbin{i+l}{i}=\qbin{w\cdot p+g-1}{r\cdot p+a}.
	\end{equation}
	Because $i+l\ge i$, Lemma \ref{thm:LemFondamental} shows that this term can be zero only when, if $g>0$, $\qbinmini{g-1}{a}$ is or, if $g=0$, when $\qbinmini{p-1}{a}$ is. But $\qbinmini{p-1}{a}$ is never zero as $a\le p-1$ and the term \eqref{eq:facteur1} is zero if and only if $0<g\le a$. The factor $[i+l+1]=[w\cdot p+g]$ will be zero when $g=0$. Therefore the factors in \eqref{eq:facteur0} are singular at $q_c$ if and only if $0\le g\le a$ and the singularity in $a_{i,j,m}$, if any, is a simple pole by Lemma \ref{thm:LemFondamental}.

	The term $[k+l+1]=[2j+1]$ is zero when (and only when) $2j+1\equiv 0\modulo p$ which is precisely the definition of $j$ being critical. In this case, the pole in the factor \eqref{eq:facteur0}, if any, is canceled by the zero in $[k+l+1]$ and the first statement follows.

	The second will be obtained if one can rule out the cases when the factor $\qbinmini{i}{k}$ provides a zero at $q_c$. Again Lemma \ref{thm:LemFondamental}, with the fact that $i\ge k$, shows that this factor is non-zero if and only if $\qbinmini{a}{d}$ is, i.e. if $d\le a$.
\end{proof}

	The previous lemma is simple. Still it can be cast in a diagrammatic version that makes the identification of singular coefficients almost trivial. Recall first that the coefficients $a_{i,j,m}$ of $z_{j,m}$ are zero when $i<j-m$. To distinguish them, we shall call {\em spurious} the coefficients $a_{i,j,m}$ with $i<j-m$ and {\em normal} the others.
\begin{defn}
	Let $p$ and $m$ be fixed with their meaning as above. The allowed values for $j$ are $m\le j\le \frac n2$. These values are organized into {\em cycles} as follows. A cycle is a set $\{j_0, j_1, \dots\}$ of a maximum of $p$ consecutive allowed values of $j$ such that $j_0$ labels a normal coefficient $a_{i,j_0,m}$, it satisfies $j_0\equiv m\modulo p$, and the $j_l=j_0+l$ are included as long as $0\le l<p$ and $j_l$ labels a normal coefficient. The {\em rightmost cycle} on the line $i$ is the unique cycle that contains at least one element and such that the next value of $j$ on its right is either larger than $\frac n2$ or corresponds to a spurious coefficient. A pair of allowed $j$ and $j'$ is {\em bound} if $j$ and $j'$ are distinct, belong to the same cycle and $j+j'\equiv p-1\modulo p$. For fixed $i$ and $m$, the labels $(a,d,g)$ and $(a',d',g')$ corresponding to $(i,j,m)$ and $(i,j',m)$ through \eqref{eq:adg} are also called {\em bound} if $j$ and $j'$ are. (Note that $a'=a$ then.)
\end{defn}

An allowed value of $j$ is critical if $2j+1\equiv 0\modulo p$. Therefore a pair $j$ and $j'$ can be bound only if they are both non-critical. Figures \ref{fig1} and \ref{fig2} give the labels $(a,d,g)$ of the normal coefficients $a_{i,j,m}$ for $n=30$, $m=9$ and $p=4$, and $n=25$, $m=\frac{13}2$ and $p=5$, respectively. The last line of Figure \ref{fig1}, now to be described, provides examples of the above definition. The critical lines are indicated by dashed vertical lines. If $p$ is odd  as in Figure \ref{fig2}, every other critical line falls between two consecutive $j$'s and every (complete) cycle contains precisely one critical $j$ and an even number of non-critical $j$'s, all of the latter bound pairwise. If $p$ is even, all critical lines avoid the $j$'s, as in Figure \ref{fig1}, {\em or} go through them. Again the number of non-critical $j$'s is even, all bound pairwise. Therefore any non-critical $j$ in a cycle of $p$ elements forms a bound pair with some element $j'$ in this cycle. The spurious ones are denoted by $\bullet$. Each column has, as header, the idempotent $z_{j,m}$, then starts by $(j-m)$ dots and goes on with the labels $(a,d,g)$ for all its (normal) coefficients $a_{i,j,m}$. The cycles are circumscribed by rounded rectangles. The heavier ones indicate the rightmost cycles which may contain less than $p$ (normal) elements. Bound pairs are easily read: They correspond to pairs $(a,d,g)$ and $(a,d',g')$ in the same cycle that lie symmetrically on each side of the critical line through the cycle closest to them. 

\begin{figure}[t!]
	\centering
	\begin{tikzpicture}[baseline={(current bounding box.center)},scale=2]
	\node at (1,10) {$z_{9,9}$};\node at (2,10) {$z_{10,9}$};\node at (3,10) {$z_{11,9}$};
	\node at (4,10) {$z_{12,9}$};\node at (5,10) {$z_{13,9}$};\node at (6,10) {$z_{14,9}$};
	\node at (7,10) {$z_{15,9}$};
	\draw[densely dashed] (1.5,10.2) -- (1.5,6.1);
	\draw[densely dashed] (3.5,10.2) -- (3.5,6.1);
	\draw[densely dashed] (5.5,10.2) -- (5.5,6.1);
	\node at (0.1,9.5) {$i=0$}; 
	\node at (1,9.5) {$(0,0,3)$};\node at (2,9.5) {$\bullet$};\node at (3,9.5) {$\bullet$};
	\node at (4,9.5) {$\bullet$};\node at (5,9.5) {$\bullet$};\node at (6,9.5) {$\bullet$};
	\node at (7,9.5) {$\bullet$};
	\draw[line width=2pt] (0.8,9.7) arc (90:270:0.2);
	\draw[line width=2pt] (0.8,9.7) -- (1.2,9.7);
	\draw[line width=2pt] (0.8,9.3) -- (1.2,9.3);
	\node at (0.1,9.0) {$i=1$}; 
	\node at (1,9.0) {$\underline{(1,0,0)}$};\node at (2,9.0) {$\underline{(1,1,1)}$};\node at (3,9.0) {$\bullet$};
	\node at (4,9.0) {$\bullet$};\node at (5,9.0) {$\bullet$};\node at (6,9.0) {$\bullet$};
	\node at (7,9.0) {$\bullet$};
	\draw[line width=2pt] (0.8,9.2) arc (90:270:0.2);
	\draw[line width=2pt] (0.8,9.2) -- (2.2,9.2);
	\draw[line width=2pt] (0.8,8.8) -- (2.2,8.8);
	\node at (0.1,8.5) {$i=2$}; 
	\node at (1,8.5) {$\underline{(2,0,1)}$};\node at (2,8.5) {$\underline{(2,1,2)}$};\node at (3,8.5) {$(2,2,3)$};
	\node at (4,8.5) {$\bullet$};\node at (5,8.5) {$\bullet$};\node at (6,8.5) {$\bullet$};
	\node at (7,8.5) {$\bullet$};
	\draw[line width=2pt] (0.8,8.7) arc (90:270:0.2);
	\draw[line width=2pt] (0.8,8.7) -- (3.2,8.7);
	\draw[line width=2pt] (0.8,8.3) -- (3.2,8.3);
	\node at (0.1,8.0) {$i=3$}; 
	\node at (1,8.0) {$\underline{(3,0,2)}$};\node at (2,8.0) {$\underline{(3,1,3)}$};\node at (3,8.0) {$\underline{(3,2,0)}$};
	\node at (4,8.0) {$\underline{(3,3,1)}$};\node at (5,8.0) {$\bullet$};\node at (6,8.0) {$\bullet$};\node at(7,8) {$\bullet$};
	\draw[line width=2pt] (0.8,8.2) arc (90:270:0.2);
	\draw[line width=2pt] (0.8,8.2) -- (4.2,8.2);
	\draw[line width=2pt] (0.8,7.8) -- (4.2,7.8);
	\draw[line width=2pt] (4.2,7.8) arc (90:270:-0.2);
	\node at (0.1,7.5) {$i=4$};
	\node at (1,7.5) {$(0,0,3)$};\node at (2,7.5) {$(0,1,0)$};\node at (3,7.5) {$(0,2,1)$};
	\node at (4,7.5) {$(0,3,2)$};\node at (5,7.5) {$(0,0,3)$};\node at (6,7.5) {$\bullet$};
	\node at (7,7.5) {$\bullet$};
	\draw (0.8,7.7) arc (90:270:0.2);
	\draw (0.8,7.7) -- (4.2,7.7);
	\draw (0.8,7.3) -- (4.2,7.3);
	\draw (4.2,7.3) arc (90:270:-0.2);
	\draw[line width=2pt] (4.8,7.7) arc (90:270:0.2);
	\draw[line width=2pt] (4.8,7.7) -- (5.2,7.7);
	\draw[line width=2pt] (4.8,7.3) -- (5.2,7.3);
	\node at (0.1,7.0) {$i=5$};
	\node at (1,7.0) {$\underline{(1,0,0)}$};\node at (2,7.0) {$\underline{(1,1,1)}$};\node at (3,7.0) {$(1,2,2)$};
	\node at (4,7.0) {$(1,3,3)$};\node at (5,7.0) {$\underline{(1,0,0)}$};\node at (6,7.0) {$\underline{(1,1,1)}$};
	\node at (7,7.0) {$\bullet$};
	\draw (0.8,7.2) arc (90:270:0.2);
	\draw (0.8,7.2) -- (4.2,7.2);
	\draw (0.8,6.8) -- (4.2,6.8);
	\draw (4.2,6.8) arc (90:270:-0.2);
	\draw[line width=2pt] (4.8,7.2) arc (90:270:0.2);
	\draw[line width=2pt] (4.8,7.2) -- (6.2,7.2);
	\draw[line width=2pt] (4.8,6.8) -- (6.2,6.8);
	\node at (0.1,6.5) {$i=6$};
	\node at (1,6.5) {$\underline{(2,0,1)}$};\node at (2,6.5) {$\underline{(2,1,2)}$};\node at (3,6.5) {$(2,2,3)$};
	\node at (4,6.5) {$(2,3,0)$};\node at (5,6.5) {$\underline{(2,0,1)}$};\node at (6,6.5) {$\underline{(2,1,2)}$};
	\node at (7,6.5) {$(2,2,3)$};
	\draw (0.8,6.7) arc (90:270:0.2);
	\draw (0.8,6.7) -- (4.2,6.7);
	\draw (0.8,6.3) -- (4.2,6.3);
	\draw (4.2,6.3) arc (90:270:-0.2);
	\draw[line width=2pt] (4.8,6.7) arc (90:270:0.2);
	\draw[line width=2pt] (4.8,6.7) -- (7.2,6.7);
	\draw[line width=2pt] (4.8,6.3) -- (7.2,6.3);
\end{tikzpicture}
	\caption{\label{fig1}The labels $(a,d,g)$ for the coefficients $a_{i,j,m}$ of the $z_{j,m}$ with $n=30$, $m=9$ and $p=4$. The vertical (dashed) lines are the critical ones. All cycles are circumscribed, the rightmost ones with a heavier line. The underlined $(a,d,g)$ correspond to singular $a_{i,j,m}$.}
\end{figure}
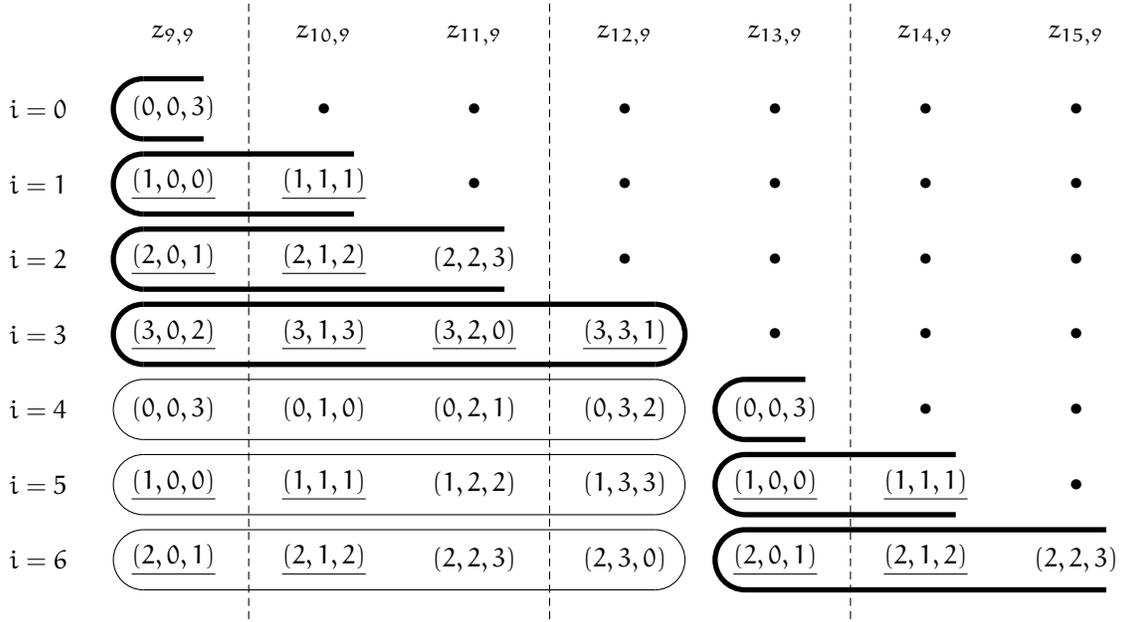

\begin{figure}[t!]
	\centering
	\begin{tikzpicture}[baseline={(current bounding box.center)},scale=2]
	\node at (1,10) {$z_{\frac{13}2,\frac{13}2}$};\node at (2,10) {$z_{\frac{15}2,\frac{13}2}$};\node at (3,10) {$z_{\frac{17}2,\frac{13}2}$};
	\node at (4,10) {$z_{\frac{19}2,\frac{13}2}$};\node at (5,10) {$z_{\frac{21}2,\frac{13}2}$};\node at (6,10) {$z_{\frac{23}2,\frac{13}2}$};
	\node at (7,10) {$z_{\frac{25}2,\frac{13}2}$};
	\draw[densely dashed] (1.5,10.2) -- (1.5,6.1);
	\draw[densely dashed] (6.5,10.2) -- (6.5,6.1);
	\draw[densely dashed] (4.0,10.2) -- (4.0,10.15);
	\draw[densely dashed] (4.0,9.85) -- (4.0,8.2);
	\draw[densely dashed] (4.0,7.85) -- (4.0,7.65);
	\draw[densely dashed] (4.0,7.35) -- (4.0,7.15);
	\draw[densely dashed] (4.0,6.85) -- (4.0,6.65);
	\draw[densely dashed] (4.0,6.35) -- (4.0,6.1);
	\node at (0.1,9.5) {$i=0$}; 
	\node at (1,9.5) {$(0,0,4)$};\node at (2,9.5) {$\bullet$};\node at (3,9.5) {$\bullet$};
	\node at (4,9.5) {$\bullet$};\node at (5,9.5) {$\bullet$};\node at (6,9.5) {$\bullet$};
	\node at (7,9.5) {$\bullet$};
	\draw[line width=2pt] (0.8,9.7) arc (90:270:0.2);
	\draw[line width=2pt] (0.8,9.7) -- (1.2,9.7);
	\draw[line width=2pt] (0.8,9.3) -- (1.2,9.3);
	\node at (0.1,9.0) {$i=1$}; 
	\node at (1,9.0) {$\underline{(1,0,0)}$};\node at (2,9.0) {$\underline{(1,1,1)}$};\node at (3,9.0) {$\bullet$};
	\node at (4,9.0) {$\bullet$};\node at (5,9.0) {$\bullet$};\node at (6,9.0) {$\bullet$};
	\node at (7,9.0) {$\bullet$};
	\draw[line width=2pt] (0.8,9.2) arc (90:270:0.2);
	\draw[line width=2pt] (0.8,9.2) -- (2.2,9.2);
	\draw[line width=2pt] (0.8,8.8) -- (2.2,8.8);
	\node at (0.1,8.5) {$i=2$}; 
	\node at (1,8.5) {$\underline{(2,0,1)}$};\node at (2,8.5) {$\underline{(2,1,2)}$};\node at (3,8.5) {$(2,2,3)$};
	\node at (4,8.5) {$\bullet$};\node at (5,8.5) {$\bullet$};\node at (6,8.5) {$\bullet$};
	\node at (7,8.5) {$\bullet$};
	\draw[line width=2pt] (0.8,8.7) arc (90:270:0.2);
	\draw[line width=2pt] (0.8,8.7) -- (3.2,8.7);
	\draw[line width=2pt] (0.8,8.3) -- (3.2,8.3);
	\node at (0.1,8.0) {$i=3$}; 
	\node at (1,8.0) {$\underline{(3,0,2)}$};\node at (2,8.0) {$\underline{(3,1,3)}$};\node at (3,8.0) {${(3,2,4)}$};
	\node at (4,8.0) {${(3,3,0)}$};\node at (5,8.0) {$\bullet$};\node at (6,8.0) {$\bullet$};\node at(7,8) {$\bullet$};
	\draw[line width=2pt] (0.8,8.2) arc (90:270:0.2);
	\draw[line width=2pt] (0.8,8.2) -- (4.2,8.2);
	\draw[line width=2pt] (0.8,7.8) -- (4.2,7.8);
%
	\node at (0.1,7.5) {$i=4$};
	\node at (1,7.5) {$\underline{(4,0,3)}$};\node at (2,7.5) {$\underline{(4,1,4)}$};\node at (3,7.5) {$\underline{(4,2,0)}$};
	\node at (4,7.5) {$(4,3,1)$};\node at (5,7.5) {$\underline{(4,4,2)}$};\node at (6,7.5) {$\bullet$};
	\node at (7,7.5) {$\bullet$};
	\draw[line width=2pt] (0.8,7.7) arc (90:270:0.2);
	\draw[line width=2pt] (0.8,7.7) -- (5.2,7.7);
	\draw[line width=2pt] (0.8,7.3) -- (5.2,7.3);
	\draw[line width=2pt] (5.2,7.3) arc (90:270:-0.2);
	\node at (0.1,7.0) {$i=5$};
	\node at (1,7.0) {${(0,0,4)}$};\node at (2,7.0) {${(0,1,0)}$};\node at (3,7.0) {${(0,2,1)}$};
	\node at (4,7.0) {$(0,3,2)$};\node at (5,7.0) {${(0,4,3)}$};\node at (6,7.0) {${(0,0,4)}$};
	\node at (7,7.0) {$\bullet$};
	\draw (0.8,7.2) arc (90:270:0.2);
	\draw (0.8,7.2) -- (5.2,7.2);
	\draw (0.8,6.8) -- (5.2,6.8);
	\draw (5.2,6.8) arc (90:270:-0.2);
	\draw[line width=2pt] (5.8,7.2) arc (90:270:0.2);
	\draw[line width=2pt] (5.8,7.2) -- (6.2,7.2);
	\draw[line width=2pt] (5.8,6.8) -- (6.2,6.8);
	\node at (0.1,6.5) {$i=6$};
	\node at (1,6.5) {$\underline{(1,0,0)}$};\node at (2,6.5) {$\underline{(1,1,1)}$};\node at (3,6.5) {${(1,2,2)}$};
	\node at (4,6.5) {$(1,3,3)$};\node at (5,6.5) {${(1,4,4)}$};\node at (6,6.5) {$\underline{(1,0,0)}$};
	\node at (7,6.5) {$\underline{(1,1,1)}$};
	\draw (0.8,6.7) arc (90:270:0.2);
	\draw (0.8,6.7) -- (5.2,6.7);
	\draw (0.8,6.3) -- (5.2,6.3);
	\draw (5.2,6.3) arc (90:270:-0.2);
	\draw[line width=2pt] (5.8,6.7) arc (90:270:0.2);
	\draw[line width=2pt] (5.8,6.7) -- (7.2,6.7);
	\draw[line width=2pt] (5.8,6.3) -- (7.2,6.3);
\end{tikzpicture}
	\caption{\label{fig2}The labels $(a,d,g)$ for $n=25$, $m=\frac{13}2$ and $p=5$. (See Figure \ref{fig1}.)}
\end{figure}
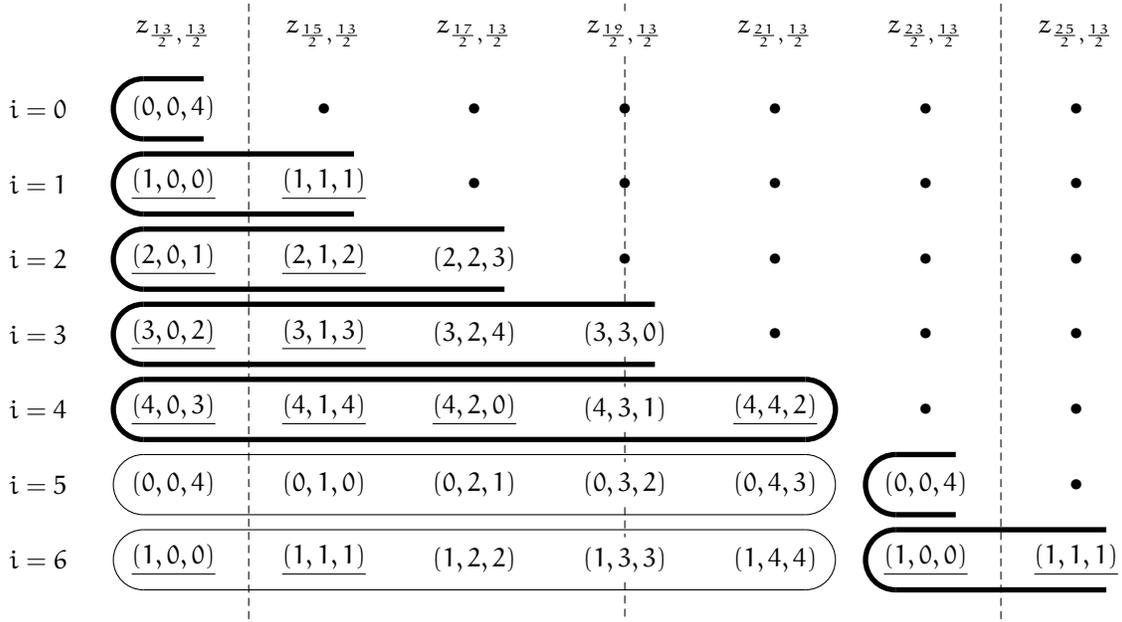

The definition \eqref{eq:adg} constrains the labels $(a,d,g)$ by
\begin{equation}
	\label{eq:constraint1}
	g-d-a\equiv 2m+1\modulo p.
\end{equation}
Moreover the condition $j+j'\equiv p-1\modulo p$ on bound pairs is equivalent to either one of
\begin{equation}
	\label{eq:constraint2}
	g\equiv a-d'\modulo p\qquad\textrm{and}\qquad g'\equiv a-d\modulo p
\end{equation}
on the labels $(a,d,g)$ and $(a,d',g')$.

\begin{lem}
	\label{thm:singu2}
	Let $p$ and $m$ be as before and suppose $(a,d,g)$ and $(a,d',g')$ (corresponding to $a_{i,j,m}$ and $a_{i,j',m}$ respectively) are bound for these $p$ and $m$. Then the three following statements are equivalent:
	\begin{enumerate}
		\item[(1)] $d, d'\le a$;
		\item[(2)] $a_{i,j,m}$ is singular at $q_c$;
		\item[(3)] $a_{i,j',m}$ is singular at $q_c$.
	\end{enumerate}
\end{lem}
\begin{proof} Since $(a,d,g)$ and $(a,d',g')$ are bound, the corresponding $j$ and $j'$ are distinct and therefore so are $d$ and $d'$.

\noindent (1) $\Rightarrow$ (2) and (3): Because $d$ and $d'$ ($\le a$) are bound, \eqref{eq:constraint2} holds. It follows from $0\le a, g\le p-1$ that $1-p\le a-g\le p-1$. If $a-g<0$, then, by \eqref{eq:constraint2}, the label $d'$ must be equal to $a-g+p$. But then $g=(a-d')+p\ge p$, a contradiction. Therefore $a-g\ge 0$. The argument is symmetric under the exchange $(a,d,g)\leftrightarrow(a,d',g')$ and therefore $a\ge g'$ as well and (1) implies (2) and (3) by the second statement of Lemma \ref{thm:singu1}.

\noindent (2) $\Leftrightarrow$ (3): Indeed, if $g\le a$ and $d\le a$, equation \eqref{eq:constraint2} gives
$$d'\equiv a-g\modulo p\qquad \textrm{and}\qquad g'\equiv a-d\modulo p$$
and therefore
$$d'=a-g\le a\qquad \textrm{and}\qquad g'=a-d'\le a$$
which imply (3), again by the second statement of Lemma \ref{thm:singu1}. Again the symmetry under $(a,d,g)\allowbreak \leftrightarrow(a,d',g')$ gives (3) $\Rightarrow$ (2).

\noindent (2) $\Rightarrow$ (1): Since (2) $\Rightarrow$ (3), the inequality $d,d'\le a$ are both automatically satisfied.
\end{proof}

The new condition (1) of the above lemma gives a useful criterion readable immediately from diagrams similar to that of Figures \ref{fig1} and \ref{fig2}. By the definition of a cycle, the label $d$ grows from $0$ to $p-1$ (if it is a complete cycle) starting from the left. The requirement that the pair $j$ and $j'$ within a cycle be bound has been described earlier. On the line $i=r\cdot p+a$, the rightmost cycle has precisely $a+1$ normal elements, the following $(p-a-1)$ ones being either spurious or not allowed. The new criterion that $d$ and $d'$ be smaller or equal to $a$ thus requires that they be among the (normal) elements of the rightmost cycle. Since, on a given line $i$, the labels $(a,d,g)$ are periodic of period $p$, the singularities of coefficients of all cycles of this line can be read from the rightmost one. 

Figures \ref{fig1} and \ref{fig2} provide examples. The labels of the singular $a_{i,j,m}$'s are underlined. All others are regular. On the top line of Figure \ref{fig1}, the label $(0,0,3)$ corresponding to $a_{0,9,9}$ is not bound as it belongs to the rightmost cycle and is alone in it. On the last line $i=6$, the coefficients $a_{6,9,9}$ and $a_{6,10,9}$ with labels $(2,0,1)$ and $(2,1,2)$ are singular because their labels belong to the rightmost cycle. However, on the same line, the label $(2,3,0)$ does not appear in the rightmost cycle ($d'=3 \not\le a=2$) and the pair $a_{6,11,9}$ and $a_{6,12,9}$ are regular at $q_c$.
\subsection{The idempotents}
A simple consequence of Lemmas \ref{thm:singu1} or \ref{thm:singu2} is that the limit of $z_{j,m}$ when $q\rightarrow q_c$ might not exist whenever $j$ is not critical and forms a bound pair with some partner $j'$. The search for new idempotents is based on the last lemma, the diagrammatic criterion discussed above and the evaluation principle.

Fix $n, m$ and $p$ with their usual meaning. Suppose, for the time being, that the number of lines $(n/2-m+1)$ of the diagram is at least $p$ and consider the rightmost cycles on the top $p$ lines of the diagram. (Figures \ref{fig1} and \ref{fig2} provide good examples of the argument that follows.) These cycles appear to the left of the diagram and again the coefficients $a_{i,j,m}$ of the projector $z_{j,m}$ form one column of the diagram. Since all non-critical $j$'s in these cycles are paired in a bound pair in at least one of the $p$ lines, the projector $z_{j,m}$ with $j\in\{j_0=m,j_1=m+1, \dots, j_{p-1}=m+p-1\}$ is either critical (and its $a_{i,j,m}$'s are all regular at $q_c$) or has at least one singular coefficient $a_{i,j,m}$ in these $p$ first rightmost cycles and this singular coefficient is paired with another singular coefficient $a_{i,j',m}$ where $j$ and $j'$ form a bound pair. (Note that if $p>n$, a projector $z_{j,m}$ with $j$ non-critical may remain regular at $q_c$ and is therefore automatically an idempotent.) Since the goal is to build well-defined projectors out of those with singular coefficients (``well-defined'' meaning with regular coefficients at $q_c$), the only hope is that the sum $z_{j,m}+z_{j',m}$, with $j$ and $j'$ bound, has regular coefficients in the top $p$ lines. Indeed, note that, if one starts reasoning with the leftmost singular idempotent labeled by $j$, one sees that its first singularity, that is the one with the smallest $i$, can be canceled only by the singularity appearing in the idempotent labeled by $j'$, its bound partner. (The possibility that all or any of the $S_i$s attached to singular $a_{i,j,m}$s be zero will be ruled out in the proof of next theorem.) The argument can then be repeated for the second leftmost singular idempotent, and so on. This possible cancellation does occur as the next theorem proves. It does not only for one of the singular coefficients of $z_{j,m}$ and $z_{j',m}$ but actually for all their singularities. Finally note that the idempotents that fall in the rightmost cycle might be regular at $q_c$ despite being non-critical. This occurs when their label $j$ fails to have a bound partner $j'$. An example occurs in Figure \ref{fig1} with $z_{15,9}$ (the last column).

\begin{thm}[Idempotents at $q$ a root of unity]\label{thm:idemRootOfUnity}
	Let $n,m$ and $p$ be as before.
	\begin{enumerate}
		\item[(1)] If $j$ is critical, then $z_{j,m}$ is an idempotent at $q_c$.
		\item[(2)] If $j$ is non-critical, falls in an incomplete rightmost cycle and does not have a bound partner $j'$, then $z_{j,m}$ is an idempotent at $q_c$.
		\item[(3)] If $j$ and $j'$ form a bound pair, the limit $z_{(j,j'),m}=\lim_{q\rightarrow q_c} (z_{j,m}+z_{j',m})$ is an idempotent at $q_c$.
		\item[(4)] If $j$ and $j'$ form a bound pair, the limit $n_{(j,j'),m}=\lim_{q\rightarrow q_c}[p]z_{j,m}$ is a nilpotent endormorphism acting non-trivially only on the submodule $z_{(j,j'),m}W_m$.
		\item[(5)] The idempotents described in (1)--(3) are orthogonal and primitive, and form a partition of unity in $\mathrm{End}_{\tl}W_m$.
	\end{enumerate}
\end{thm}

\begin{proof} 
Before constructing the idempotents, we note that all those described in statements (1)--(3) arise as limits of linear combinations of idempotents for generic $q$ (in cases (1) and (2), the limit is trivial). Thus, if the limits exist, the limiting objects inherit the properties of being idempotent and orthogonal from the generic case. For example, if $\lim_{q\rightarrow q_c} (z_{j,m}+z_{j',m})$ exists, then
$$(z_{(j,j'),m})^2=\big(\lim_{q\rightarrow q_c} (z_{j,m}+z_{j',m})\big)^2=\lim_{q\rightarrow q_c} (z_{j,m}+z_{j',m})^2=\lim_{q\rightarrow q_c} (z_{j,m}+z_{j',m})=z_{(j,j'),m}$$
where the third equality follows from Theorem \ref{thm:idqgene}. Orthogonality is obtained similarly. Finally, since all the $z_{j,m}$'s of the generic case appear either alone in cases (1) and (2) or in a bound pair in (3), the sum of their limits will have the same trace as that of the generic case, that is $\dim W_m$, and they will form a partition of unity. Clearly the two claims in statement (5) of orthogonality and that they form a partition of unity follow from the fact that the objects described in (1)--(3) are non-singular. The proof that they are will be the first step. The primitivity of the idempotents will require statement (4) whose proof will then be next. Statement (5) will appear as a consequence of (1)--(4).

Since, by Lemma \ref{thm:singu2}, the coefficients of $z_{j,m}$ are regular at $q_c$ for $j$ satisfying either (1) or (2), then it is well-defined and an idempotent by Theorem \ref{thm:idqgene}.

For the case (3), consider a bound pair $(j,j')$. One of the coefficients $a_{i,j,m}$ in the sum $z_{j,m}=\sum_i a_{i,j,m}S_i$ has then a simple pole at $q_c$. Does this imply that one of the matrix elements of $z_{j,m}$ also has such a pole? Even though the set $\{S_0,S_1, \dots, S_{n/2-m}\}$ is known to span $\mathrm{End}_{\tl}W_m$ for all $q$, it has not been proved to be a basis when $q$ is a root of unity (at least to our knowledge). A different argument is therefore needed to show that $z_{j,m}$ is indeed singular at $q_c$. Recall that $S_i=(S^-)^i(S^+)^i/[i]!^2$ and consider the linear combination of elements of the spin basis of $W_m$ obtained by acting with $S_i$ on $\ket{v}=\ket{--\dots -++\dots +}\in W_m$, the element starting with $(n/2-m)$ signs ``$-$''. Any coefficient in this linear combination is either $0$ or a polynomial in $q^{\pm \frac12}$ as can be seen from \eqref{eq:sp}. And such a coefficient is non-zero only if the vector it multiplies has all its ``$-$'' signs at the leftmost positions like $\ket{v}$, except for at most $i$ signs. Since $j$ and $j'$ form a bound pair, there exists $I\leq n/2-m$ such that $a_{I,j,m}$ is singular at $q_c$ and none of the $a_{i,j,m}$ with $i>I$ is. Let us study the matrix element $\bra{w}z_{j,m}\ket{v}$ where $\ket v$ is as above and
\begin{equation*}
\ket w=\ket{\ \underbrace{- - \dots -}_{n/2-m-I}\ \underbrace{+ + \dots +}_{n/2+m}\ \underbrace{- - \dots -}_{I}\ }.
\end{equation*}
The contribution to this matrix element of $a_{i,j,m}\bra w S_i \ket v$ is regular at $q_c$ for $i>I$ by definition of $I$. Moreover $\bra w S_i \ket v=0$ for all $i<I$ since, in that case, $S_i$ cannot change the positions of that many ``$-$'' signs to go from $\ket v$ to $\ket w$. For $\bra w S_I \ket v$, equation \eqref{eq:sp} leads to
\begin{align*}
\bra w S_I \ket v
  & = \bra w \frac{(S^-)^I}{[I]!}\frac{(S^+)^I}{[I]!} \ket v \\
  & = \bra w 
     \sum_{j_1<j_2<\dots<j_I} S_{j_1}^-S_{j_2}^-\dots S_{j_I}^-
     \sum_{i_1<i_2<\dots<i_I} S_{i_1}^+S_{i_2}^+\dots S_{i_I}^+ \ket v,
\end{align*}
of which sum only one term survives, that with $i_\ell=n/2-m-I+\ell$ and $j_\ell=n-I+\ell$. Therefore $\bra w S_I \ket v$ is a power of $q^{\pm\frac12}$ and is non-zero at any $q\in\mathbb C^\times$. (One computes easily that $\bra w S_I \ket v=q^{-I(n/2-m-I)}$.) The only non-zero term in the sum $\bra w z_{j,m} \ket v=\bra w\sum_i a_{i,j,m}S_i\ket v$ occurs for $i=I$ and $\bra w z_{j,m} \ket v$ is singular at $q_c$ because $a_{I,j,m}$ is.

Statement (3) will then be proved if the limit $\lim_{q\rightarrow q_c}(a_{i,j,m}+a_{i,j',m})$ exists for all $i$, $0\le i\le \frac n2-m+1$. Since $j$ and $j'$ are bound, there exists some $i$'s in this range such that $a_{i,j,m}$ is singular at $q_c$. By Lemma \ref{thm:singu2} this occurs if and only if $a_{i,j',m}$ is also singular. For such an $i$, Lemma \ref{thm:LemFondamental} has established that the poles of $a_{i,j,m}$ and $a_{i,j',m}$ at $q_c$ are simple. The goal will therefore be to show that the sum of the two singular coefficients behaves close to $q_c$ as
\begin{equation}
	\label{eq:desiredForm}
	a_{i,j,m}+a_{i,j',m}\sim f(q)/[p]
\end{equation}
and that $f(q)\rightarrow 0$ as $q\rightarrow q_c$. (Recall that all coefficients are products of $q$-numbers and therefore, if $f(q)$ vanishes at $q_c$, it must have a zero of integer degree at this point.) The computation is straightforward, though somewhat messy. Some observations are needed before proceeding.

Suppose that $a_{i,j,m}$ and $a_{i,j',m}$ are singular with $j$ and $j'$ a bound pair. Without loss of generality we assume $j<j'$. Let
\begin{alignat}{4}
	\label{eq:adgprime}
	i       		& = r\cdot p+a,   	& \qquad   i'=i 			& = r\cdot p+a, 		\notag\\
	j-m   	& = u\cdot p+d, 	& \qquad   j'-m      		& = u' \cdot p+d', 		\\
	i+j+m+1 	& = w\cdot p+g, 	& \qquad\qquad  i+j'+m+1   & = w'\cdot p+g'	\notag
\end{alignat}
be their labels with the usual assumption $0\le a, d, d', g, g'\le p-1$. Since every cycle starts with a $j_0-m\equiv 0\modulo p$, the requirement $j'>j$ implies $d'>d$ and $u'=u$. Then equation \eqref{eq:constraint2} gives
$$g\equiv a-d'\modulo p\qquad \textrm{and}\qquad g'\equiv a-d\modulo p$$
and, because $d,d'\le a$ by Lemma \ref{thm:singu2},
\begin{equation}
	\label{eq:lesDeuxG}
	g=a-d'\qquad\textrm{and}\qquad g'=a-d.
\end{equation}
Therefore $j'>j$ implies $g'-g=d'-d>0$. Moreover $p-1\ge j'-j=(w'-w)\cdot p+(g'-g)>0$ forces $w'=w$. Finally equation \eqref{eq:lesDeuxG} together with the fact that $d,d'\le a$ implies that $g,g'\le a$. Since clearly $i+j+m+1\ge i$, then $w$ and $w'$ must be strictly larger that $r$.

We now simplify the various factors of
$$a_{i,j,m}=(-1)^{i+j-m}\qbin{i}{j-m}\qbin{i+j+m}{i}^{-1}\frac{[2j+1]}{[i+j+m+1]}$$
using Lemma \ref{thm:LemFondamental}. The first $q$-binomial is
\begin{equation*}
	\qbin{i}{j-m}=\qbin{r\cdot p+a}{u\cdot p+d}=q^{e_1}\bin{r}{u}\qbin{a}{d}
\end{equation*}
where $e_1=(a-d)up+(r-u)(up+d)p$, and similarly for the triplet $(i,j',m)$. Note that, for the latter, only $d$ is changed for its primed partner $d'$. The $q$-binomial stemming from this term and its primed partner are
\begin{align*}
	\qbin{a}{d}	& = \frac{[a]!}{[d]!\,[a-d]!} = \frac{[a]!}{[d']!\,[a-d]!}\times\bigl([d']\dots [d+1]\bigr)		\\[1mm]
	\qbin{a}{d'}	& = \frac{[a]!}{[d']!\,[a-d']!} = \frac{[a]!}{[d']!\,[a-d]!}\times\bigl([a-d]\dots [a-d'+1]\bigr).
\end{align*}
From now on, the dots in an expression of the form $[x]\dots [y]$ mean $\prod_{z=y}^x[z]$ and assume $x\ge y$.

The next pair of factors, $[2j+1]$ and $[2j'+1]$, is the simplest as
$$[2j'+1]=-[2j+1]$$
since $j'+j\equiv p-1\modulo p$. The final factor is
$$\qbin{i+j+m}{i}^{-1}\frac1{[i+j+m+1]}=\qbin{(w-1)\cdot p+(p+g-1)}{r\cdot p+a}^{-1}\frac1{[(w-1)\cdot p+p+g]}.$$
The peculiar writing of $i+j+m$ as $(w-1)\cdot p+(p+g-1)$ was chosen because $g<a$ but $(p+g-1)\ge a$. Since $w$ is strictly larger than $r$ (see discussion following equation \eqref{eq:lesDeuxG}) and at least $1$ when $a_{i,j,m}$ is singular, the binomial $\left(\begin{smallmatrix}w-1\\ r\end{smallmatrix}\right)$ is always non-zero. Lemma \ref{thm:LemFondamental} then gives that this term behaves as
$$q^{-e_2}\begin{pmatrix}w-1 \\ r\end{pmatrix}^{-1}\frac{[a]!\,[p+g-a-1]!}{[p+g]!}=q^{-e_2}\begin{pmatrix}w-1 \\ r\end{pmatrix}^{-1}\frac{[a]!}{[p+g]\dots [p+g-a]}\qquad\textrm{for $q$ close to $q_c$}
$$
and with $e_2=(p+g-1-a)rp+(w-r-1)(rp+a)p-(w-1)p$. Because 
\begin{equation}
	\label{eq:beaucoupInegal}
	p+g'>p+g\ge p\ge p+g'-a>p+g-a, 
\end{equation}
the $q$-numbers in the above expression and those in its primed version can be written as
\begin{align*}
	\frac{[a]!}{[p+g]\dots [p+g-a]}	& = \frac{[a]!}{[p+g']\dots [p+g-a]}\times\bigl([p+g']\dots [p+g+1]\bigr)		\\[2mm]
	\frac{[a]!}{[p+g']\dots [p+g'-a]}	& = \frac{[a]!}{[p+g']\dots [p+g-a]}\times\bigl([p+g'-a-1]\dots [p+g-a]\bigr)
\end{align*}
Due to the inequalities \eqref{eq:beaucoupInegal}, the common denominator $[p+g']\dots [p+g-a]$ contains $(g'-g)+a+1\ge3$ terms and one of them is $[p]$. It thus contains the only singular term of both $a_{i,j,m}$ and $a_{i,j',m}$. The sum of these two coefficients can therefore be factorized as 
\begin{equation}
	\label{eq:petita}
	a_{i,j,m}+a_{i,j',m}\sim (-1)^{i+j-m}q^{e_1-e_2}\begin{pmatrix}r\\ u\end{pmatrix} 
\begin{pmatrix} w-1\\ r\end{pmatrix}^{-1}\times \frac{[a]!}{[d']!\,[a-d]!}\cdot [2j+1]\cdot \frac{[a]!}{[p+g']\dots
[p+g-a]}\times \big( \dots \big)
\end{equation}
where
\begin{align}
	\label{eq:petitb}
	\big( \dots \big) = [d']	& \dots [d+1]\times [p+g']\dots [p+g+1]										\notag\\
						& - (-1)^{g'-g}q^{p((r-2u)(d'-d)-r(g'-g))}[a-d]\dots [a-d'+1]\times [p+g'-a-1]\dots [p+g-a].
\end{align}
All the factors in front of $\big( \dots \big)$ in equation \eqref{eq:petita} exist, that is are finite, except for the singular $[p]$ in the last denominator. This equation \eqref{eq:petita} is therefore of the desired form \eqref{eq:desiredForm}. Using
\begin{equation*}
	[p+x]=q^p[x]\qquad\textrm{and}\qquad [-x]=-[x],
\end{equation*}
equation \eqref{eq:lesDeuxG} and $g'-g=d'-d$, we can write $\big( \dots\big )$ as
\begin{equation*}
	[d']\dots [d+1]\times q^{p(g'-g)}[g']\dots [g+1]-(-1)^{g'-g}[g']\dots [g+1]\times (-1)^{d'-d}q^{p(d'-d)}[d']\dots [d+1]
\end{equation*}
which is clearly zero. Therefore $\lim_{q\to q_c}a_{i,j,m}+a_{i,j',m}$ exists for all $i$. Statement (3) follows.

We turn now to statement (4). Let $(j,j')$ be a bound pair with $j<j'$ and let $n_{(j,j'),m}$ be defined as
\begin{equation}\label{eq:fjjpm}n_{(j,j'),m}=\lim_{q\rightarrow q_c}[p]z_{j,m}.\end{equation}
This $n_{(j,j'),m}$ is a non-zero element of $\mathrm{End}_{\tl}W_m$. To see this, recall that $z_{j,m}=\sum_i a_{i,j,m}S_i$. The endomorphisms $S_i$ have polynomial matrix elements in the variables $q$ and $q^{-1}$ (or $q^{\frac12}$ and $q^{-\frac12}$) and they are therefore regular at $q=q_c$. The limit $\lim_{q\rightarrow q_c} [p]a_{i,j,m}$ always exists as the $a_{i,j,m}$'s are either regular at $q=q_c$ (and then the limit is zero) or have a simple pole (and then the limit is the non-zero residue). The proof of statement (3) has established that $z_{j,m}$ and $z_{j',m}$ have at least one singular matrix element when $j$ and $j'$ form a bound pair and the above limit is thus a non-zero endomorphism in $\text{End}_{\tl}W_m$.

The endomorphism $n_{(j,j'),m}$ is nilpotent:
$$n_{(j,j'),m}^2=\big(\lim_{q\rightarrow q_c}[p]z_{j,m}\big)^2=
\lim_{q\rightarrow q_c}([p])^2z_{j,m}z_{j,m}=\lim_{q\rightarrow q_c}([p])^2z_{j,m}=0$$
because $z_{j,m}$ is an idempotent in a neighborhood of $q_c$ and, again, the pole in $z_{j,m}$ is simple. The endomorphism $n_{(j,j'),m}$ acts as zero on all subspaces $z_{k,m}W_m$ for the $k$'s of cases (1) and (2) and $z_{(J,J'),m}W_m$ for $(J,J')$ a bound pair distinct from $(j,j')$. For example
$$n_{(j,j'),m}z_{(J,J'),m}=\big(\lim_{q\rightarrow q_c}[p]z_{j,m}\big)\big(\lim_{q\rightarrow q_c}z_{J,m}+z_{J',m}\big)=\lim_{q\rightarrow q_c}[p]z_{j,m}(z_{J,m}+z_{J',m})=\lim_{q\rightarrow q_c}[p]\cdot 0=0$$
by the orthogonality of the idempotents for generic $q$. Since it is non-zero, $n_{(j,j'),m}$ acts non-trivially only on $z_{(j,j'),m}W_m$.

The idempotents (1)--(3) are orthogonal and, if $(j,j')$ is bound, then $z_{(j,j'),m}$ and $n_{(j,j'),m}$ are linearly independent in $\mathrm{End}_{\tl} W_m$, since one is idempotent and the other nilpotent. Therefore statements (1)--(4) provide $(n/2-m+1)$ linearly independent endomorphisms on $W_m$. Since the $\{S_0, S_1, \dots, S_{n/2-m}\}$ spans $\mathrm{End}_{\tl} W_m$ by the Schur-Weyl duality at a root of unity, then $\dim \mathrm{End}_{\tl} W_m\leq n/2-m+1$. The idempotents and nilpotents of (1)--(4) form then a basis of $\mathrm{End}_{\tl} W_m$, and so does the set $\{S_0, S_1, \dots, S_{n/2-m}\}$.

Can $n_{(j,j'),m}$ be used to decompose $z_{(j,j'),m}$ as a sum of non-zero orthogonal idempotents? If this is possible, one of them is of the form $\alpha z_{(j,j'),m}+\beta n_{(j,j'),m}$ with $\beta\neq 0$. But then $\big(\alpha z_{(j,j'),m}+\beta n_{(j,j'),m}\big)^2=\alpha^2 z_{(j,j'),m}+2\alpha\beta n_{(j,j'),m}$ and the requirement that it be idempotent forces $\alpha=\alpha^2$ and $2\alpha\beta=\beta$ which do not have a solution with $\beta$ non-zero. The idempotent $z_{(j,j'),m}$ is therefore primitive and statement (5) follows.
\end{proof}

It is an interesting exercise to count how many idempotents Theorem \ref{thm:idemRootOfUnity} has identified. Since the statement of the theorem is for a fixed $m$, the exercise amounts to identifying the number of values of $m$ for which a given idempotent exists. The critical $j$ (case (1)) provides the simplest example. If $j$ is critical for the root $q_c$ under study, the projector $z_{j,m}$ computed for the generic case remains regular. There are therefore submodules $z_{j,m}W_m$ for all $m\le j$, that is $(2j+1)$ in total. Recall that such a subspace $z_{j,m}W_m$ first appears for $m=j$. The central element $F_n$ can take at most two distinct values if $j$ is critical and $F_n$ acting on the principal or standard modules takes these values only if the module has a critical $j$ as index. Recalling that $W_j$ contains a submodule isomorphic $\mathcal V_j$ and using a recursive argument to rule out other $\mathcal V_{j'}$ with critical $j'$ if necessary, one concludes that the module $z_{j,j}W_j$ is isomorphic to $\mathcal V_j$ and is irreducible. Since the $\uq$-modules with the corresponding value for the Casimir given by Corollary \ref{coro:fn} are the $\mathcal M_j$ which are cyclic for the extended algebra (they are generated by the highest weight vector), then all subspaces $z_{j,m}W_m$ for $-j\le m\le j$ are isomorphic to $\mathcal V_j$. Note that the multiplicity of the standard $\mathcal V_j\cong \mathcal P_j$ just found coincides with that in equation \eqref{eq:GY} (second sum).

The multiplicities of the submodules upon which the idempotents of cases (2) and (3) project need to be computed simultaneously. If $(j',j)$ is a bound pair with $j>j'$, the submodule $z_{(j',j),m}W_m$ will be denoted by $p_j^m$. All modules $p_j^m$ with the same $j$ have the same dimension, again because of the argument on the continuity of the trace as a function of $q$ introduced at the beginning of Theorem \ref{thm:idemRootOfUnity}'s proof. We shall drop the index $m$ on $p_j^m$ as we are interested in the total number of modules. (We need to stress that the following argument does not prove that $p_j^m\cong p_j^{m'}$.) Case (2) may occur only for the last element $j_l$ of a non-critical orbit $orb_j$. All other elements $j_i\in orb_j$ will bound to either $j_{i-1}$ or $j_{i+1}$ depending on $m$. Indeed, due to Lemma \ref{thm:singu2}, a projector $z_{j,m}$ will always have at least one singular coefficient in its expansion \eqref{eq:idem2} if $j$ is not the last element of its orbit. If again the elements of the orbit $orb_j$ are labeled by $j_1<j_2<\dots <j_l$, then the submodule $p_j$ with the smallest $j$ is $p_2$ as we chose to keep the largest element of the bound pair $(j',j)$ as label. A submodule $p_{j_2}$ will appear for each $m\le j_1$ and the number $\#p_{j_2}$ of such modules will be $(2j_1+1)$. If there is a third element $j_3$ in $orb_j$, the number $\#p_{j,3}$ is computed as follows. Again a bound pair $(j_2,j_3)$ may occur only if $m\le j_2$. However $j_2$ may be bound to either $j_1$ or $j_3$ and therefore $\#p_{j_3}=(2j_2+1)-\#p_{j_2}$. The same argument can be repeated to give $\#p_{j_{i+1}}=(2j_i+1)-\#p_{j_i}$ for all $i=1,2,\dots,l-1$. Equation \eqref{eq:orbit} may be used to express the $j_i$ in terms of $j_1$ or $j_l$ (for example $j_{2k+1}=kp+j_1$ and $j_{2k}=kp-j_1-1$) and the solution of the recursion is found to be $\#p_{j_i}=(i-1)(ip-2j_i-1)$. These are therefore the multiplicities associated with the idempotents of case (3). Those of case (2) occur only when the last element $j_l$ of the orbit is not bound to the previous element $j_{l-1}$ and the multiplicity is $((2j_l+1)-\#p_{j_l})$. The dimension of the submodules $z_{j_l,m}W_m$ is that of the $\mathcal V_{j_l}$, that is $\Gamma^{(n)}_{j_l}$ and we will denote these modules by $v_{j_l}$. Their multiplicity is therefore $\#v_{j_l}=l(2j_l+1-(l-1)p)$. Again we note that the multiplicities $\#p_{j_i}$ and the dimension $\dim p_{j_i}=\Gamma^{(n)}_{j_i}$ coincide with those of the first sum in \eqref{eq:GY} and those for the $v_j$ with those of the third. Theorem \ref{thm:idemRootOfUnity} thus reproduces the multiplicity of \eqref{eq:GY}.

When all the $p_j$'s and the $v_j$'s have distinct dimensions (which is common), the above argument also proves that $p_j^m\cong \mathcal P_j$ for all $m$ and $v_j\cong \mathcal V_j$. Since the idempotents were our goal, we do not provide finer arguments that would resolve cases with coincidences among the dimensions.

\section{Concluding remarks}

The explicit expressions (\ref{eq:coefExpr}--\ref{eq:idem2}) for the idempotents $z_{j,m}$ for generic $q$ and the linear combinations (1--3) of Theorem \ref{thm:idemRootOfUnity} that survive at $q$ a root of unity are the main result of this paper. The rules established in Lemma \ref{thm:singu2} and in Theorem \ref{thm:idemRootOfUnity} allows for an easy graphical decomposition of $W_m$. When $q$ is generic, the result is simple: $W_m\cong\oplus_{m\le j\le n/2}\mathcal V_j$ (Corollary \ref{cor:isoMathcalV}). When $q$ is a root of unity associated with $p$, that is, $p$ is the smallest integer such that $q^{2p}=1$, then the decomposition of $W_m$ as a $\tl$-module is read from the $n$-th line of the Bratteli diagram with the critical lines drawn (corresponding to the solutions of $2j+1\equiv 0\textrm{\,mod\,} p$). Only the $j\ge m$ play a role, either forming bound pairs or remaining alone. Any of these $j$'s appear only once in the linear combinations of $z_{j,m}$. For example Figure \ref{fig:BratDiag3} shows the decomposition of $W_3\subset \otimes^{20}\mathbb{C}^2$ when $q$ is a root associated with $p=5$. Starting at $m=3$ and proceeding to the right, all pairs symmetric with respect to critical lines are bound: First the pair $(3,6)$, then $(4,5)$ and, since $5$ has already been paired, the last pair $(9,10)$. The critical $j=7$ correspond to a regular idempotent at this $q$ and $j=8$, the last element of the orbit of $\text{orb}_{j=1}$, remains unbound. Therefore the indecomposable modules are $z_{(3,6),3}W_3$, $z_{(4,5),3}W_3$, $z_{(9,10),3}W_3$, $z_{7,3}W_3$ and $z_{8,3}W_3$ to be put in relation through \eqref{eq:GV} with the explicit decomposition 
$W_3=
\mathcal P_5\oplus
\mathcal P_6\oplus
\mathcal P_{10}\oplus
\mathcal V_7\oplus
\mathcal V_8
$.
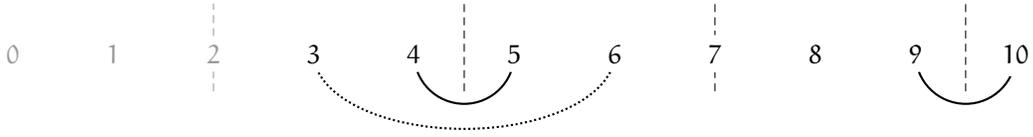
\begin{figure}[h!tb]
	\centering
	\begin{tikzpicture}[baseline={(current bounding box.center)},every node/.style={fill=white,circle,inner sep=2pt},scale=2/3]
	\begin{scope}[thick]
		\draw (10,0) arc [start angle=0, end angle=-180, x radius=1, y radius=1] -- (8,0);
		\draw (20,0) arc [start angle=0, end angle=-180, x radius=1, y radius=1] -- (18,0);
	\end{scope}
	\begin{scope}[thick,densely dotted]
		\draw (12,0) arc [start angle=0, end angle=-180, x radius=3, y radius=1.5] -- (6,0);
	\end{scope}
	\draw[densely dashed, draw opacity=0.4] (4,-0.75) -- (4,-0.35);
	\draw[densely dashed, draw opacity=0.4] (4,0.35) -- (4,1);
	\draw[densely dashed] (9,-0.75) -- (9,1.);
	\draw[densely dashed] (14,-0.75) -- (14,1.);
	\draw[densely dashed] (19,-0.75) -- (19,1.);
\pgfsetfillopacity{0.4}
	\node at (0,0) {$0$};
	\node at (2,0) {$1$};
	\node at (4,0) {$2$};
\pgfsetfillopacity{1}	
	\node at (6,0) {$3$};
	\node at (8,0) {$4$};
	\node at (10,0) {$5$};
	\node at (12,0) {$6$};
	\node at (14,0) {$7$};
	\node at (16,0) {$8$};
	\node at (18,0) {$9$};
	\node at (20,0) {$10$};
\end{tikzpicture}
	\caption[Bratteli diagram for $q$ root of unity]{\label{fig:BratDiag3}The decomposition of $W_3\subset \otimes^{20}\mathbb{C}^2$ when $p=5$. Only linear combinations of projectors $z_{j,m}$ with $j\ge m=3$ may occur.}
\end{figure}

Projectors are basic tools in physical applications of representation theory. The earliest example is of course the Clebsch-Gordan coefficients and the $3$-$j$ symbols arising in the quantum theory of angular momentum. The present expressions (\ref{eq:coefExpr}--\ref{eq:idem2}) are the analog for the Temperley-Lieb algebra. They are easily coded in any symbolic manipulation program and may be useful to investigate properties of the Hamiltonian $H_{XXZ}=\sum_{1\le i\le n-1}E_i$ or of any other object expressed in terms of the generators of the Temperley-Lieb algebra $\tl$. The relatively simple form of these idempotents raises the question of its possible extension for the duality between $\mathrm{U}_q\mathfrak{sl}_n$ and the Hecke algebra $H_n(q)$ whose solution could be of significant physical relevance.

\section*{Acknowledgements}

We thank Paul Martin for directing us to his result \cite{Martin1992}, and Alexi Morin-Duchesne for a careful reading of the manuscript. This work is supported by the Canadian Natural Sciences and Engineering Research Council 
(Y.~S.-A.).

%
%
\appendix
\section{Basic identities}
\begin{prop}\label{prop:prodqnb}For $b,c\in\Z$
\begin{equation*}
		[b][c] = [b+c-1]+[b+c-3]+\dots+[c-b+3]+[c-b+1].
\end{equation*}
\end{prop}
\begin{proof}Simply expand
$$(q^{b-1}+q^{b-3}+\dots+q^{-(b-3)}+q^{-(b-1)})(q^c-q^{-c})/(q-q^{-1}).$$
\end{proof}

\begin{prop}
	\label{prop:sumqnb}
	\begin{enumerate}
		\item[(a)] For $0\leq l<k$ and $2j,k\in\N$,
		\begin{equation*}
			A_l = \sum_{r=0}^l(-1)^r\frac{[2j+k+r]!}{[r]!\,[2j+r+1]!\,[k-r]!}= (-1)^l\frac{[2j+k+l+1]!}{[2j+l+1]!\,[k-l-1]!\,[l]!\,[2j+k+1]\,[k]}.
		\end{equation*}
		\item[(b)] For $0\leq l<i$ and $2m,i\in\N$,
		\begin{equation*}
			\label{eq:sumqnbb}
			B_l = \sum_{r=0}^l(-1)^r\frac{[2m+r]!\,[2m+2r+1]}{[r]!\,[i-r]!\,[2m+r+i+1]!} = (-1)^l\frac{[2m+l+1]!}{[2m+i+l+1]!\,[i-l-1]!\,[l]!\,[i]}.
		\end{equation*}
	\end{enumerate}
\end{prop}
\begin{proof} Both relations are proved by induction. The proofs are similar and we give that for $A_l$. When $l=0$, the result above follows easily. Assuming the relation for $A_l$, one gets
	\begin{align*}
		A_{l+1} 	& = A_l + (-1)^{l+1}\frac{[2j+k+l+1]!}{[l+1]!\,[2j+l+2]!\,[k-l-1]!}	\\[3mm]
				& = (-1)^l\frac{[2j+k+l+1]!}{[2j+l+2]!\,[k-l-1]!\,[l+1]!\,[2j+k+1]\,[k]}\,\bigl([l+1]\,[2j+l+2]-[k]\,[2j+k+1]\bigr).
	\end{align*}
	But $[l+1]\,[2j+l+2] = [2j+2l+2]+[2j+2l]+\dotsb+[2j+2]$ and $[k]\,[2j+k+1] = [2j+2k]+[2j+2k-2]+\dotsb+[2j+2]$ according to Proposition \ref{prop:prodqnb}. Thus
	\begin{align*}
		[l+1]\,[2j+l+2]-[k]\,[2j+k+1]	& = -\bigl([2j+2k]+[2j+2k-2]+\dotsb+[2j+2l+4]\bigr)	\\
								& = -[k-l-1]\,[2j+k+l+2]
	\end{align*}
again by the same proposition. Therefore the expression becomes
	\begin{align*}
		A_{l+1}	& =  (-1)^l\frac{[2j+k+l+1]!}{[2j+l+2]!\,[k-l-1]!\,[l+1]!\,[2j+k+1]\,[k]}\,\bigl(-[k-l-1]\,[2j+k+l+2]\bigr) \\[3mm]
				& = (-1)^{l+1}\frac{[2j+k+l+2]!}{[2j+l+2]!\,[k-l-2]!\,[l+1]!\,[2j+k+1]\,[k]}.
	\end{align*}
\end{proof}

\begin{lem}
	\label{thm:LemFondamental}
	Let $q_c$ be a root of unity and $p\ge 2$ the smallest integer such that $q_c^{2p}=1$. If $k,k'\ge0$ and $a, a'\ge 0$, then 
	\begin{equation}
		\label{eq:LemFondamental}
		\qbin{kp+a}{k'p+a'}\sim q^{(a-a')k'p+(k-k')(k'p+a')p}\bin{k}{k'}\qbin{a}{a'}, \qquad\textrm{as\ }q\rightarrow q_c.
	\end{equation}
Moreover, if the rational function $\qbinmini{a}{a'}$ of $q$ has a zero at $q_c$, this zero is of degree one.
\end{lem}
\noindent This lemma is sometimes called the $q$-Lucas Theorem (\citet{Desarmenien1982, Sagan1992}). Its use will be mostly for $a$ and $a'$ in the range $0\le a, a'\le p-1$. Still the more general form is useful.

The next result is an immediate consequence of the defining relation of $\uq$.
\begin{prop}\label{prop:SpmThru}For $k\in\Z$ and $n\in\N$
\begin{equation}
\label{eq:SpmThru}(S^\pm)^n\,\bigl[2S^z+k\bigr] = \bigl[2S^z+k\mp2n\bigr]\,(S^\pm)^n, \qquad \textrm{and therefore}\qquad \bigl[S_n,\bigl[2S^z+k\bigr]\bigr] = 0.
\end{equation}	
\end{prop}

\begin{prop}\label{lem:sksl} If $l\geq k$, the restriction of the product $S_kS_l$ to $W_m\subset\cn$, $m\ge 0$, is given by
\begin{equation*}
	\left.S_kS_l\right|_{W_m} = \sum_{i=0}^k\qbin{l+i}{k}\qbin{l+i}{l}\qbin{2m+k+l}{k-i}\,\left.S_{l+i}\right|_{W_m}.
\end{equation*}
\end{prop}
\begin{proof}
	Using equation \eqref{eq:SpmRenorm}, we first expand the product:
	\begin{equation*}
		S_kS_l 	 = (S^-)^{(k)}(S^+)^{(k)}(S^-)^{(l)}(S^+)^{(l)} 
		= \sum_{i=0}^k(S^-)^{(k)}\qbin{2S^z+l-k}{i}(S^-)^{(l-i)}(S^+)^{(k-i)}(S^+)^{(l)}.
	\end{equation*}
	The last three divided powers may be commuted past the $q$-binomial using the result $(S^\pm)^{(k)}[2S^z+r]=[2S^z+r\mp2k](S^\pm)^{(k)}$ 
by Proposition \ref{prop:SpmThru}
or by simply evaluating $S^z$ when the above expression acts on $W_m$. We find by restricting to $W_m$ (the restriction symbol is omitted):
	\begin{align*}
		S_kS_l 	& = \sum_{i=0}^k(S^-)^{(k)}(S^-)^{(l-i)}(S^+)^{(k-i)}(S^+)^{(l)}\qbin{2m+k+l}{i}							\\[1mm]
				& = \sum_{i=0}^k(S^-)^{(k+l-i)}(S^+)^{(k+l-i)}\,\frac{\bigl([k+l-i]!\bigr)^2}{[k]!\,[l-i]!\,[k-i]!\,[l]!}\,\qbin{2m+k+l}{i}	\\[1mm]
				& = \sum_{i=0}^k\qbin{k+l-i}{k}\qbin{k+l-i}{l}\qbin{2m+k+l}{i}\,S_{k+l-i}.
	\end{align*}
	A change of the index of summation gives the statement.
\end{proof}

\begin{lem}
	\label{lem:commutS}
	Let $m,n\in\N$. Then
	\begin{equation}
		\label{eq:commutS}
		\bigl[S_m,S_n\bigr] = 0.
	\end{equation}
\end{lem}
\begin{proof}
	Let us first show that
	\begin{equation}
		\label{eq:commutS1Sn}
		\bigl[S^-S^+,(S^-)^n(S^+)^n\bigr] = 0.
	\end{equation}
	For $n=0$ and $1$, the previous holds trivially. Now suppose that $\left[S^-S^+,(S^-)^{n-1}(S^+)^{n-1}\right]=0$. Using equation \eqref{eq:SpmRenorm}, we get
\begin{align*}
	(S^-)^n(S^+)^n 	& = S^-(S^-)^{n-1}S^+(S^+)^{n-1} = S^-\left(S^+(S^-)^{n-1}-[n-1][2S^z-n+2](S^-)^{n-2}\right)(S^+)^{n-1}	\\
					& = (S^-S^+)\bigl((S^-)^{n-1}(S^+)^{n-1}\bigr)-[n-1]\bigl((S^-)^{n-1}(S^+)^{n-1}\bigr)[2S^z-n]
\end{align*}
by Proposition \ref{prop:SpmThru}. Now each pair of $S$'s between parentheses commute with $S^-S^+$ by the induction hypothesis.

We now increase the exponent of the first term in the commutator, again by induction. To compute $\bigl[(S^-)^m(S^+)^m,(S^-)^n(S^+)^n\bigr]$, we express $(S^-)^m(S^+)^m$ using the above relation with $n\rightarrow m$. Again induction and equation \eqref{eq:commutS1Sn} show that each pair of $S$'s commute with $(S^-)^n(S^+)^n$.
The result of the proposition follows by dividing by $([m]!)^2\,([n]!)^2$.
\end{proof}

\begin{prop}
	\label{prop:Sraction}
	Let $j$ and $j'$ be the labels of the ``tall'' and ``short'' towers of the $\uq$-module $U_{j,j'}$. The action of an element $S_r$ of $\mathrm{End}_{\tl}W_m$, with $0\leq m\leq j'$, on the vectors $\ket{j,m}$ and $\ket{j',m}$ is given by
	\begin{equation}
		\label{eq:Sraction1}
		S_r\ket{j,m} = \qbin{j+m+r}{r}\qbin{j-m}{r}\ket{j,m},\qquad\text{and}
	\end{equation}
	\begin{multline*}
		S_r\ket{j',m} = \qbin{j'+m+r}{r}\qbin{j'-m}{r}\ket{j',m}+{}	\\
		\frac1{[r]!}\qbin{j-m}{r}\sum_{i=1}^r\qbin{j-m-i}{j'-m-i+1}\frac{[j'+m+i-1]!}{[j'+m]!}\frac{[j+m+r]!}{[j+m+i]!}\ket{j,m}.
	\end{multline*}
In particular,
\begin{align*}
		S_1\ket{j,m} &= [j+m+1][j-m]\ket{j,m},\qquad \text{ and} \\ 
		S_1\ket{j',m}&= [j'+m+1][j'-m]\ket{j',m}+\frac{[j-m]!}{[j-j'-1]!\,[j'-m]!}\ket{j,m}.
	\end{align*}
\end{prop}
\begin{proof}
	First, we need the action of $(S^\pm)^{(r)}$ on the vectors $\ket{j,m}$ and $\ket{j',m}$. Using repeatedly \eqref{eq:actionSpm}, we obtain
	\begin{align}
		(S^\pm)^{(r)}\ket{j,m}	& = \frac{(S^\pm)^{r-1}}{[r]!}\,[j\pm m+1]\,\ket{j,m\pm1}			\notag\\
							& = \dotsb 												\notag\\
							& = \frac{[j\pm m+1][j\pm m+2]\dotsm[j\pm m+r]}{[r]!}\,\ket{j,m\pm r}	\notag\\
							& = \qbin{j\pm m+r}{r}\ket{j,m\pm r}.\label{eq:Spm}
	\end{align}
	Similarly, using the first relation of \eqref{eq:SpmSmallTower}, we find that $(S^-)^{(r)}$ on $\ket{j',m}$ also acts diagonally:
	\begin{equation*}
		(S^-)^{(r)}\ket{j',m} = \qbin{j'-m+r}{r}\ket{j',m-r}.
	\end{equation*}
	The non-diagonal action of $(S^+)^{(r)}$ on $\ket{j',m}$ is found using the second relation of \eqref{eq:SpmSmallTower} and also \eqref{eq:Spm}:
	\begin{equation*}
		(S^+)^{(r)}\ket{j',m} = \qbin{j'+m+r}{r}\ket{j',m+r}+\frac1{[r]!}\sum_{i=1}^r\qbin{j-m-i}{j'-m-i+1}\frac{[j'+m+i-1]!}{[j'+m]!}\frac{[j+m+r]!}{[j+m+i]!}\ket{j,m+r}.
	\end{equation*}
	The action of $S_r=(S^-)^{(r)}(S^+)^{(r)}$ on the tower vectors follows from those equations.
\end{proof}
\noindent Note that, if $q$ is generic, no coupling between two towers occurs, and the action of $S_r$ is thus diagonal and given by \eqref{eq:Sraction1}.

\bibliographystyle{myplainnat}
\bibliography{end29}

\end{document}